\documentclass[preprint,12pt]{elsarticle}
\usepackage[utf8]{inputenc}
\usepackage{todonotes}
\usepackage{graphicx,color}
\usepackage{tikz}
\usepackage{amsmath,mathrsfs,amsfonts,latexsym,amssymb,verbatim}
\usepackage{algpseudocode}
\usepackage{cclicenses}
\usepackage{hyperref}
\usepackage{cleveref}
\usepackage{microtype}%if unwanted, comment out or use option "draft"
\usepackage{relsize}
\newcommand{\VL}[1]{}
\newcommand{\VC}[1]{#1}
\definecolor{brightcerulean}{RGB}{0,123,167}
\definecolor{brightother}{RGB}{167,123,0}
\definecolor{brightorange}{RGB}{204,82,0}
\definecolor{darkgreen}{RGB}{0,128,43}
\newcommand{\anonymous}[1]{{{#1}^\circ}}
\newcommand{\acom}[1]{{\color{brightcerulean}}}
\newcommand{\changes}[1]{{#1}}
\newcommand{\changesII}[1]{#1}

\newcommand{\setG}{{\cal X}}
\newcommand{\setN}{\mathbb{N}}

\newcommand{\assymext}[1]{#1^{\square}}

\newcommand{\reversepath}[1]{#1^{-1}}
\newcommand{\bint}{\{l,r\}^*}
\newcommand{\intersectn}{\wedge}
\newcommand{\dom}{\textrm{dom}}
\newcommand{\projection}[1]{#1^\lrcorner}

\def\X{\mathcal{X}}
\def\Y{\mathcal{Y}}
\def\IM{\{lm,rm\}^*}

\newcommand{\invstate}{\mathrm{invisible}}

\newtheorem{conjecture}{Conjecture}
\newdefinition{remark}{Remark}
\newdefinition{example}{Example}
\newdefinition{definition}{Definition}
\newtheorem{theorem}{Theorem}
\newtheorem{proposition}{Proposition}
\newtheorem{lemma}{Lemma}
\newtheorem{corollary}{Corollary}
\newproof{proof}{Proof}

\def\N{\mathbb{N}}

\newcommand{\bdot}{\mathsmaller{\mathsmaller{\mathsmaller{\bullet}}}}
\begin{document}
\begin{frontmatter}
\title{Size-varying reversible causal graph dynamics}
\author[inst1]{Pablo Arrighi}

\author[inst2]{Amélia Durbec\corref{cor1}}

\author[inst3]{Aurélien Emmanuel}

\affiliation[inst1]{organization={Université Paris-Saclay, Inria, CNRS, LMF}, city={Gif-sur-Yvette}, postcode={91190}, country={France}}

\affiliation[inst2]{organization={Université Paris-Saclay, CEA, List},city={Palaiseau}, postcode={91120}, country={France}}

\affiliation[inst3]{organization={Université d’Orléans, LIFO EA 4022}, city={Orléans}, postcode={45067},country={France}}
%\maketitle              % typeset the header of the contribution
%

% Author macros::begin %%%%%%%%%%%%%%%%%%%%%%%%%%%%%%%%%%%%%%%%%%%%%

\begin{abstract}
Consider a network that evolves according to a reversible, nearest neighbours dynamics. {\em Is the dynamics allowed to vary the size of the network?} On the one hand it seems that, being the principal carriers of information, nodes cannot be destroyed without jeopardising bijectivity. On the other hand, there are plenty of bijective functions from the set of graphs to the set of graphs that are non-vertex-preserving. The question has been settled negatively---for three different reasons. Yet, in this paper we do obtain reversible local node creation/destruction---in three relaxed settings, whose equivalence we prove for robustness. We motivate our work both by theoretical computer science considerations (reversible computing, cellular automata extensions) and theoretical physics concerns (basic formalisms towards discrete quantum gravity).
%\medskip
%\noindent {\bf Keywords.} Reversible Causal Graph Dynamics, Reversible Cellular Automata, Lattice Gas Automata, Causal Dynamical Triangulations, Spin networks, invertible, one-to-one.
\end{abstract}

\end{frontmatter}

\section{Introduction}\label{sec:motivations}

Cellular Automata (CA) consist in a $\mathbb{Z}^n$ grid of identical cells, each of which may take a state in $\Sigma$. Thus the configurations are in $\Sigma^{\mathbb{Z}^n}$. The next state of a cell is given by applying a fixed local rule $f$ to the cell and its neighbours, synchronously and homogeneously across space. CA thus have a number of physics-like symmetries: shift-invariance (the dynamics acts everywhere and everywhen the same) and causality (information has a bounded speed of propagation). They constitute one of the most established models of computation that accounts for Euclidean space: they are widely used to model spatially-dependent computational problems (self-replicating machines, synchronization\ldots), and multi-agents phenomena (traffic jams, demographics\ldots). But their origin lies in Physics, where they are constantly used to model waves or particles (e.g. as numerical schemes for Partial Differential Equations). %In fact they do have a number of in-built physics-like symmetries: shift-invariance (the dynamics acts everywhere the same) and causality (information has a bounded speed of propagation). 

Since both quantum and classical mechanics are reversible, it was natural to endow CA with this other, physics-like symmetry. The study of Reversible CA (RCA) was further motivated by the promise of lower energy consumption in reversible computation. RCA have turned out to have an elegant mathematical theory, which relies on a topological characterization in order to prove for instance that the inverse of a CA is a CA \cite{Hedlund}---which clearly is non-trivial due to \cite{KariRevUndec}.
Another fundamental property of RCA is that they can be expressed as a finite-depth circuits of local reversible permutations or `blocks' \cite{KariBlock,KariCircuit,Durand-LoseBlock}. 

Causal Graph Dynamics (CGD) \cite{ArrighiCGD,ArrighiIC,ArrighiCayleyNesme,MartielMartin,Maignan} are a twofold extension of CA. First, the underlying grid is extended to arbitrary bounded-degree graphs.  Informally, this means that each vertex of a graph $G$ may take a state among a set $\Sigma$, so that configurations are in $\Sigma^{V(G)}$, whereas edges dictate the locality of the evolution: the next state of a vertex $v$ depends only upon the subgraph $G_u^r$ induced by the vertices lying at graph distance at most $r$ of $u$. Second, the graph itself is allowed to evolve over time. Informally, this means that configurations are in the union of $\Sigma^{V(G)}$ for every possible bounded-degree graph $G$, i.e. $\bigcup_G\Sigma^{V(G)}$. This leads to a model where the local rule $f$ is applied synchronously and homogeneously on every possible sub-disk of the input graph, thereby producing small patches of the output graphs, whose union constitutes the output graph. Figure \ref{fig:CGDIdea} illustrates the concept.
\begin{figure}[h]
\begin{center}
\includegraphics[scale=.25]{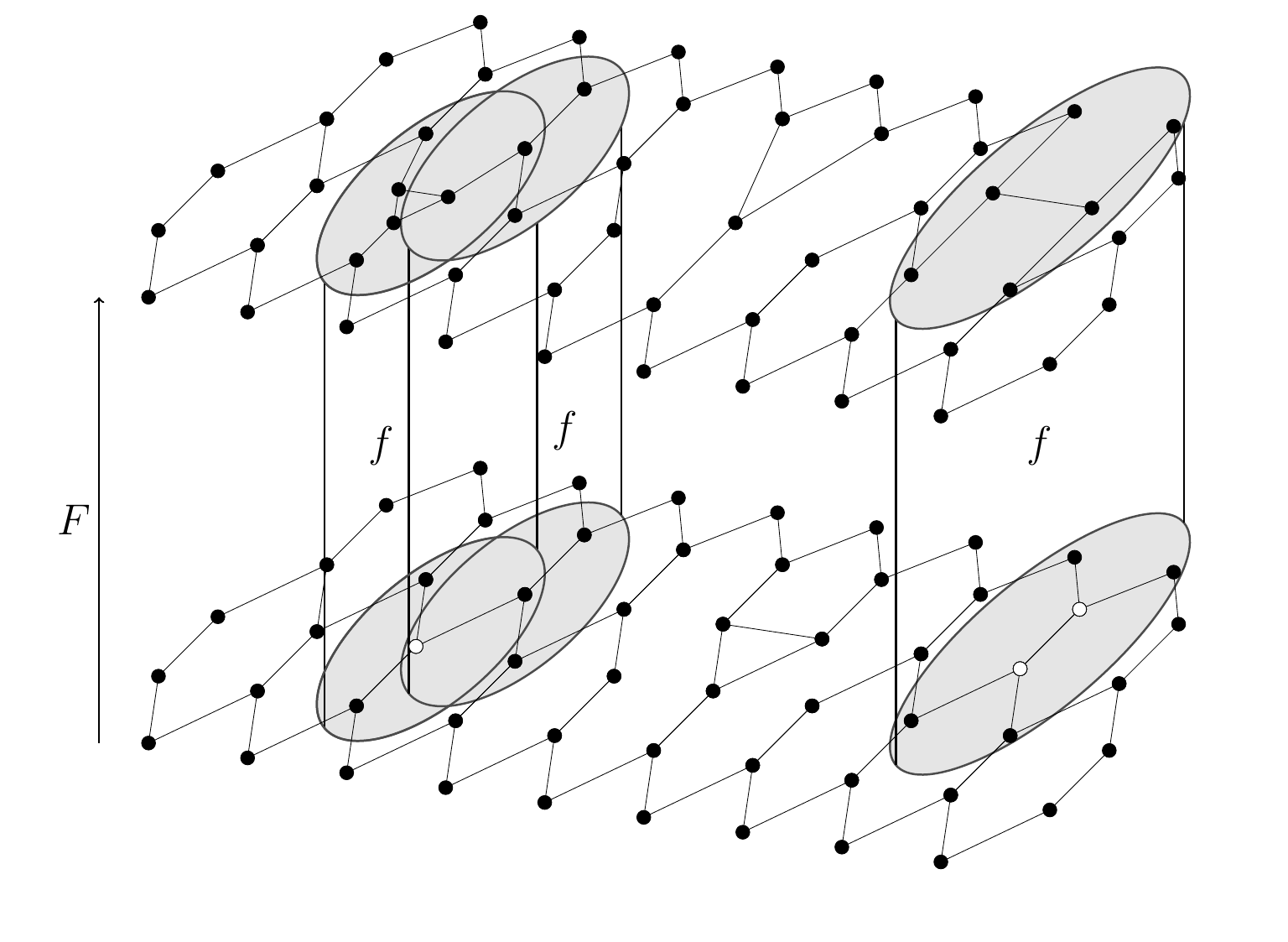}
\end{center}
\caption{\label{fig:CGDIdea} {\em Informal illustration of Causal Graph Dynamics.}
The entire graph evolves into another according to a global function $F$. But this evolution is causal (information propagates at a bounded speed) and homogeneous (same causes lead to same effects). This has been shown to be equivalent to applying a local function $f$ to every sub-disk of the input graphs, producing small output graphs whose union make up the output graph.}
\end{figure}
CGD were motivated by the countless situations featuring nearest-neighbours interactions with time-varying neighbourhood (e.g. agents exchange contacts, move around\ldots). Many existing models (of complex systems, computer processes, biochemical agents, economic agents, social networks\ldots) fall into this category, thereby generalizing CA for their specific sake (e.g. self-reproduction as in \cite{TomitaSelfReproduction}, discrete general relativity \`a la Regge calculus \cite{Sorkin}, etc.). CGD are a theoretical framework, for these models. Some graph rewriting models, such as Amalgamated Graph Transformations \cite{BFHAmalgamation} and Parallel Graph Transformations \cite{EhrigLowe,TaentzerHL}, also work out rigorous  ways of applying a local rewriting rule synchronously throughout a graph, albeit with a different, category-theory-based perspective, of which the latest and closest instance is \cite{Maignan}. In \cite{ArrighiRCGD,ArrighiBRCGD} one of the authors studied CGD in the reversible regime, i.e. Reversible CGD. Specific examples of Reversible CGD had been described in \cite{MeyerLGA,MeyerLove}. 

From a theoretical Computer Science perspective, the point was to generalize RCA theory to arbitrary, bounded-degree, time-varying graphs. Indeed the two main results in \cite{ArrighiRCGD,ArrighiBRCGD} were the generalizations of the two above-mentioned fundamental properties of RCA. However, the results were limited to (almost--) vertex--preserving CGD. We show that this limitation can be lifted.

From a mathematical perspective, questions related to the bijectivity of CA over certain classes of graphs (more specifically, whether pre-injectivity implies surjectivity for Cayley graphs generated by certain groups \cite{Bartholdi,CeccheriniEden,Gromov}) have received quite some attention. The present paper on the other hand provides a context in which to study ``bijectivity of CA over time-varying graphs''. We answer the question: {\em Is it the case that bijectivity necessarily rigidifies space (i.e. forces the conservation of each vertex)?} Our analysis pinpoints the assumptions that lead to this rigidification---and how to circumvent them.

From a theoretical physics perspective, the question whether the reversibility of small scale physics (quantum mechanics, Newtonian mechanics), can be reconciled with the time-varying topology of large scale physics (relativity), is a major challenge. This paper provides a rigorous discrete, toy model where reversibility and time-varying topology coexist and interact---in a way which does allow for space expansion. In fact these results pave the way for Quantum Causal Graph Dynamics \cite{ArrighiQCGD,ArrighiQNT} allowing for vertex creation/destruction---which in turn could provide a rigorous basic formalism to use in Quantum Gravity \cite{QuantumGraphity1,QuantumGraphity2}.

\section{The conflict between reversibility and node creation/destruction}
\noindent {\bf The question.} Consider a network that evolves reversibly, according to nearest neighbours interactions. {\em Can its dynamics create/destroy nodes?}

\noindent {\bf Issue 1.} Because the network evolves according to nearest neighbours interactions only, the same local causes must produce the same local effects. In other words, if the neighbourhood of a node $u$ looks the same as that of a node $v$, then the same must happen at $u$ and $v$. Therefore the names of the nodes must be irrelevant to the dynamics. Surely the most natural way to formalize this invariance under isomorphisms is as follows. Let $F$ be the function from graphs to graphs that captures the time evolution; we require that for any renaming $R$, $F\circ R=R\circ F$. But it turns out that this commutation condition forbids node creation, even in the absence of any reversibility condition---as proven in \cite{ArrighiCGD}. Intuitively, say that a node $u\in V(G)$ creates a node $u'\in V(G')$ through $F$, and consider an $R$ that just interchanges the $u'$ name for some fresh name $v'$. Then $F(RG)=F(G)$, which has no $v'$, differs from $RF(G)$, which has a $v'$. 

\noindent {\bf Issue 2.} The above issue can be fixed by making it explicit that new names are constructed from the locally available ones (e.g. $u'$ from $u$ in the above example), so that renaming the new names (e.g. $u'$ into $v'$ through some $R'$) necessarily implies having renamed the available ones (e.g. $u$ into $v$ through $R$). Then invariance under isomorphisms is formalized by requiring that for any renaming $R$, there exists $R'$, such that $F\circ R=R'\circ F$. But it turns out that this conjugation condition, taken together with reversibility, still forbids node creation, as proven in \cite{ArrighiIC}. To get a taste of the difficulty, say that a node $u$ creates two nodes $u.l$ and $u.r$. Then $F^{-1}$ should merge these back into a single node $u$. However, we expect $F^{-1}$ to have the same conjugation property that for any renaming $S$, there exists $S'$, such that $F^{-1}\circ S=S'\circ F^{-1}$. Consider an $S$ that leaves $u.l$ unchanged, but renames $u.r$ into some fresh $v'$. What will be the name of the merger between $u.l$ and $v'$ through $F^{-1}$, now? What should $S'$ do upon $u$ in order to obtain that name? Generally speaking, node creation between $G$ and $F(G)$ increases the naming space and endangers the bijectivity that should hold between $\{RG\}$ the set of renamings of $G$ and $\{RF(G)\}$ the set of renamings of $F(G)$. 

\noindent {\bf Issue 3.}  Both the above no-go theorems rely on naming issues. In order to bypass them, one may drop names altogether, and work with graphs modulo isomorphisms. Doing this, however, is quite inconvenient. Basic statements such as ``the neighbourhood of $u$ determines what will happen at $u$''---needed to formalize the fact the network evolves according to nearest-neighbours interactions---are no longer possible if we cannot speak of $u$.\\
Still, having chosen networks that are not mere graphs (edges are between the ports of the nodes) we can designate a node relative to another by giving a path from one to the other (the successive ports that lead to it). It then suffices to have one privileged pointed vertex acting as `the origin', to be able to designate any vertex relative to it\footnote{Having an origin is also mandatory for defining the Gromov-Hausdorff-Cantor metric, which allows for a topological characterization of these dynamics in the style of Curtys-Hedlund for Cellular Automata.}. Then, the invariance under isomorphisms is almost trivial, as nodes have no name. The one thing that remains to enforce is invariance under shifting the origin. Namely, if $X_u$ stands for $X$ with its origin shifted along the path $u$, then there must exist some successor function $R_X:V(X)\longrightarrow V(F(X))$ such that $F(X_u)=F(X)_{R_X(u)}$. But it turns out that this seemingly mild condition, when taken together with reversibility, again forbids node creation but for a finite number of graphs---as was proven in \cite{ArrighiRCGD}.\\
Intuitively, node creation between $X$ and $F(X)$ augments the number of ways in which the graph can be pointed at, i.e. the number of possible origins. This, except in a handful of cases\footnote{Sometimes it can happen that $F(X)$ has a symmetry $F(X)_{u'}=F(X)_{v'}$ with $u'\neq v'$, in which case both $u'$ and $v'$ can be created by a single node $u$. These symmetries, however, are global properties which the local rule cannot see. In order to warranty reversibility, the dynamics must therefore be vertex-preserving except in a finite number of cases.}, again endangers the bijectivity that should hold between the sets of shifts $\{X_u\}_{u\in V(X)}$ and $\{F(X)_{u'}\}_{u'\in V(F(X))}$. 

\noindent {\bf Three solutions and a plan.} In \cite{MeyerLGA}, Hasslacher and Meyer (HM) \ref{ex:HM} describe a great example of a nearest-neighbours driven size-varying dynamics. The HM example consists of particles moving around a circle, with collisions causing the circle to shrink or grow, according to the way in which particles cross. It is the long-term behaviour of the HM example which makes it so interesting. Consider a circular graph picked at random: one would expect that the circle length will erratically grow and shrink, behaving mostly like a random walk, albeit ultimately periodic. This would be in line with the intuition of there being a 50\% chance for the circle to grow or shrink at each particle encounter, as well as the fact that the number of shrinking graphs must counterbalance the number of growing graphs. Yet, most of the time, the circles end up steadily growing toward infinity, thereby inducing an `arrow of time without past hypothesis', as we recently established in \cite{TimeArrow}.\\
The point is that the HM example is clearly non-vertex-preserving, but it nevertheless seems reversible, in some sense which was left informal in \cite{MeyerLGA} and demands formalization.\\
A first approach in order to simulate the HM example (or any other size-varying dynamics) by means of a Reversible CGD, would be to embed it within a strictly reversible, vertex-preserving dynamics---where each `visible' node of the network is equipped with its own reservoir of `invisible' nodes---in which it can tap in order to create a visible node. For this scheme to iterate, and for the created nodes to be able to create nodes themselves, it is convenient to shape the reservoirs as infinite binary trees. Since the reservoirs are and stay the same everywhere, they can then be abstracted away from the model. The obtained relaxed setting thus circumvents the above three issues. Section \ref{sec:IMCGD} presents this solution.

A second, perhaps more direct approach to formalizing size-varying dynamics such as the HM example, is to work with pointed graphs modulo just when they are useful, e.g. for stating causality, and to drop the pointer everywhen else, e.g. for stating reversibility. This relaxed setting reconciles reversibility and local creation/destruction---it can be thought of as a direct response to Issue 3. Section \ref{sec:ACGD} presents this solution. We also prove that this second solution is as powerful as the first, in the sense that they simulate each other.

A third approach is to work with standard, named graphs.  Remarkably, it turns out that naming our nodes within the algebra of variables over everywhere-infinite binary trees directly resolves Issue 2. Section \ref{section:ncgd} presents this solution. 
Again we prove that this third solution is as powerful as the first, in the sense of reciprocal simulation.

It therefore appears that the question whether reversibility allows for local node creation/destruction, is formalism-dependent. But then, how sensitive to formalism can it be? \\ 
Fortunately, we were able to prove that the three proposed relaxed settings are essentially equivalent (except for small graphs behaviours), as shown by means of simulation theorems. Thus, we have reached a robust formalism allowing for both the features. Section \ref{sec:RCGD} recalls the definitions and results that constitute our point of departure. Section \ref{sec:conclusion} summarizes the contributions and perspectives. {\em This paper is the journal version of an extended abstract \cite{ArrighiCreation}, providing some corrections, as well as full-blown proofs and details.}

\begin{figure}[h]
\begin{center}
\includegraphics[scale=2]{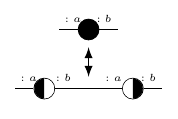}
\end{center}
\caption{\label{fig:HMScattering} {\em Hasslacher-Meyer collision step}}
\end{figure}
\begin{example}[Hasslacher-Meyer]\label{ex:HM}
The Hasslacher-Meyer dynamics consists of alternating two steps: 1) An advection step which translates particles according to their orientation 2) A collision step as shown in figure \ref{fig:HMScattering}, creating or destroying vertices according to the distance between the two particles.
\end{example}

\section{In a nutshell\thinspace: Reversible Causal Graph Dynamics }\label{sec:RCGD}

{\em The following provides an intuitive introduction to Reversible CGD. A thorough formalization was given in \cite{ArrighiCayleyNesme}, and is reproduced in \ref{app:formalism} for convenience}.

\noindent {\bf Networks.} Whether for CA over graphs \cite{PapazianRemila}, multi-agent modeling \cite{Danos200469} or agent-based distributed algorithms \cite{Chalopin}, it is common to work with graphs whose nodes have numbered neighbours. Thus our 'graphs' or networks are the usual, connected, undirected, possibly infinite, bounded-degree graphs, but with a few additional twists:
\begin{itemize}
\item[$\bullet$] The set $\pi$ of available ports to each vertex is finite.
\item[$\bullet$] The vertices are connected through their ports: an edge is an unordered pair $\{u : a, v : b\}$, where $u,v$ are vertices and $a,b\in \pi$ are ports. Each port is used at most once per node: if both $\{u: a,v: b\}$ and $\{u: a, w: c\}$ are edges, then $v=w$ and $b=c$. As a consequence, the degree of the graph is bounded by $|\pi|$.
\item[$\bullet$] The vertices are given labels taken in the finite set $\Sigma$, so that they may carry an internal state just like the cells of a CA. 
\item[$\bullet$] This labelling is partial, so that we may express our partial knowledge about part of a graph. 
\end{itemize} 
The set of all graphs (see Figure \ref{fig:graphs}$(a)$) having ports $\pi$, vertex labels $\Sigma$ is denoted $\mathcal{G}_{\Sigma,\pi}$. \vspace{6pt}\\
\noindent {\bf Compactness.}  There are two main approaches to CA. The one with a local rule is the most constructive, but CA can also be defined in a more topological way as being exactly the shift-invariant continuous functions from $\Sigma^{\mathbb{Z}^n}$ to itself, with respect to the Gromov-Hausdorff-Cantor metric. Through a compactness argument, the two approaches are equivalent. This topological approach carries through to CA over graphs, and simplifies the proofs a great deal. The compact metric space of graphs is obtained by dropping names right after having introduced a pointer \cite{ArrighiCayleyNesme}:
\begin{itemize}
\item[$\bullet$] The graphs have a privileged pointed vertex playing the role of an origin, so that any vertex can be referred to relative to the origin, via a sequence of ports that lead to it. 
\item[$\bullet$] The pointed graphs are considered modulo isomorphism, so that only the relative position of the vertices can matter.
\end{itemize} 
The set of all pointed graphs modulo (see Figure \ref{fig:graphs}$(c)$) is denoted $\mathcal{X}_{\Sigma,\pi}$.\\
If, instead, we drop the pointers but still take equivalence classes modulo isomorphism, we obtain just graphs modulo, aka `anonymous graphs'. The set of all anonymous graphs (see Figure \ref{fig:graphs}$(d)$) is denoted $\mathcal{\anonymous{X}}_{\Sigma,\pi}$.\vspace{6pt}\\
Again a thorough formalization of these graphs was given in \cite{ArrighiCayleyNesme}, and is reproduced in \ref{app:graphs} for the sake of mathematical rigor. For the sake of this paper, however, Figure \ref{fig:graphs} summarizes what there is to know about the definition of pointed graphs modulo.\vspace{6pt}\\
\begin{figure}
\begin{center}
\includegraphics[scale=.9,clip=true,trim=0cm 0.22cm 0cm 0.1cm]{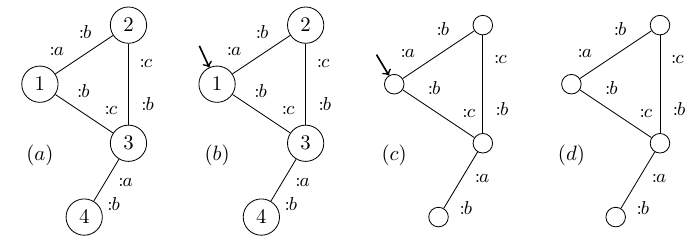}
\end{center}
\caption{\label{fig:graphs} {\em The different types of graphs.} (a) A graph $G$. (b) A pointed graph $(G,1)$. (c) A pointed graph modulo $X$. (d) An anonymous graph $\anonymous{X}$.}
\end{figure}

\noindent {\bf Paths and vertices.} Over pointed graphs modulo isomorphism, vertices no longer have a unique identifier, which may seem impractical when it comes to designating a vertex. Fortunately, any vertex of the graph can be designated by a sequence of ports in $\Pi^*$ (where $\Pi = \pi^2$)
that leads from the origin to this vertex. 
%If two paths, $u$ and $v$ designate the same vertex, $u\equiv v$.
For instance, say two vertices are designated by paths $u$ and $w$ and say that there is an edge $e=\{u: a,w: b\}$. Then, $w$ can be designated by the path $u.ab$, where ``$.$'' stands for the word concatenation. %and $u.ab\equiv w$. 
The origin is designated by $\varepsilon$.\vspace{6pt}\\

\noindent {\bf Operations over graphs.} Given a pointed graph modulo $X$, $X^r$ denotes the sub-disk of radius $r$ around the pointer. %, i.e. the subgraph of X induced by the vertices at distance at most $r$ of the origin.  
%we need to be able to take the subdisk of radius $r$ around its pointer, leading to $X^r$. 
%$X_u$ is the graph modulo obtained by moving  the pointer along the path $u$. 
%We need to be able to move the pointer of $X$ along a path $u$, leading to $Y=X_u$. We need to be able to move it back where it was before, leading to $X=Y_{\overline{u}}$. We can move the pointer, and then take the subdisk, which we write $X_u^r$. A thorough formalization of generalized Cayley graph operations can be found in \cite{ArrighiCayleyNesme}. For the sake of this paper, Figure \ref{fig:operations} illustrates the operations. 
The pointer of $X$ can be moved along a path $u$, yielding $X'=X_u$. The pointer can be moved back where it was before, yielding $X=X'_{\reversepath{u}}$, where $\reversepath{u}=b_{n-1}a_{n-1}\ldots b_0a_0$ denotes the reverse of the path $u=a_0b_0\ldots a_{n-1}b_{n-1}$. We use the notation $X_u^r$ for $(X_u)^r$ i.e., first the pointer is moved along $u$, then the sub-disk of radius $r$ is taken.
For the sake of this paper, however, Figure \ref{fig:operations} illustrates the operations.
\VC{\begin{figure}[h]
\begin{center}
 \includegraphics[scale=0.5]{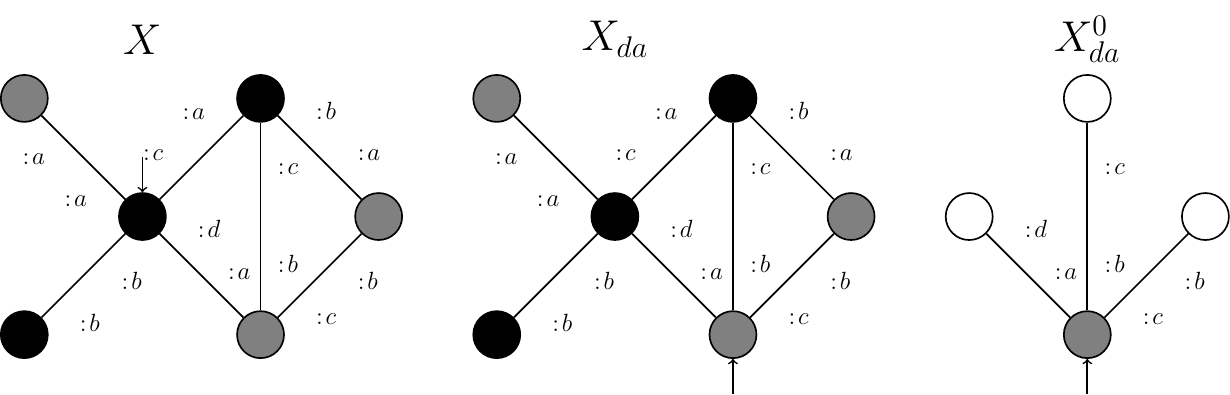}
 \end{center}
  \caption{\label{fig:operations} {\em Operations over pointed graphs modulo.} The pointer of $X$ is shifted along edge $da$, yielding $X_{da}$, and then the disk of radius $0$ around the pointer, yielding $X^0_{da}$.}
\end{figure}}

\noindent {\bf Causal Graph Dynamics.} We will now recall their topological definition. It is important to provide a correspondence between the vertices of the input pointed graph modulo $X$, and those of its image $F(X)$, which is the role of $R_X$:
\begin{definition}[Dynamics]\label{def:dynamicsmodulo}
A dynamics $(F,R_{\bullet})$ is given by
\begin{itemize}
\item[$\bullet$] a function $F\colon{\cal X}_{\Sigma,\pi}\to{\cal X}_{\Sigma,\pi}$;
\item[$\bullet$] a set of functions $R_{\bullet}=\{R_X: V(X) \to V(F(X))\ | \ X\in {\mathcal{X}}_{\Sigma,\pi}\}$.
\end{itemize}
\end{definition}
Next, continuity is the topological way of expressing causality:

\begin{definition}[Gromov-Hausdorff-Cantor metrics]\label{def:metric}
Consider the function
\begin{align*}
d:{\cal X}_{\Sigma, \pi}\times{\cal X}_{\Sigma, \pi} &\longrightarrow {\mathbb R}^+\\
(X,Y)&\mapsto d(X,Y)=0\quad \textrm{if }X=Y\\	
(X,Y)&\mapsto d(X,Y)=1/2^r\quad \textrm{otherwise}
\end{align*}
where $r$ is the minimal radius such that $X^r \neq Y^r$.\\
The function $d(.,.)$ is such that for $\epsilon>0$ we have (with $r=\lfloor - \log_2(\epsilon)\rfloor$):
$$d(X,Y)<\epsilon \Leftrightarrow X^r = Y^r.$$
It defines an ultrametric distance \cite{ArrighiCayleyNesme}.
\end{definition}

\begin{definition}[Continuity]\label{def:continuitymodulo}
A dynamics $(F,R_{\bullet})$ is said to be {\em continuous} if and only if for any $X$ and $m$, there exists $n$, such that 
$$\bullet\ (F(X))^m=(F(X^n))^m\qquad \bullet\ \dom\,R_{X}^m\subseteq V(X^n)\textrm{ and }R_{X}^m=R_{X^n}^m$$
where $R_{X}^m$ denotes the partial map obtained as the restriction of $R_X$ to the co-domain $(F(X))^m$, using the natural inclusion of $(F(X))^m$ into $F(X)$.
\end{definition}
Notice that the second condition states the continuity of $R_{X}$, for all $X$ in $\X$. A key point is that by compactness, continuity entails uniform continuity, meaning that $n$ does not depend upon $X$---so that the above really expresses that information has a bounded speed of propagation.\\
We now express that the same causes lead to the same effects:
\begin{definition}[Shift-invariance]
A dynamics $(F,R_{\bullet})$ is said to be {\em shift-invariant} if for every $X$,  $u\in V(X)$, and $v\in V(X_u)$, 
$$\bullet\ F(X_u)=F(X)_{R_X(u)}\qquad \bullet\ R_X(u.v)=R_X(u).R_{X_u}(v)$$
\end{definition}

Finally, we demand that graphs do not expand in an unbounded manner:
\begin{definition}[Boundedness]\label{def:boundednessmodulo}
A dynamics $(F,R_{\bullet})$ % from ${\cal X}_{\Sigma,\pi}$ to ${\cal X}_{\Sigma,\pi}$ 
is said to be {\em bounded} if there exists a bound $b$ such that for any $X$ and any $w\in V(F(X))$, there exists $u\in \textrm{Im}(R_X)$ and $v\in V(F(X)_{u}^b)$ 
such that $w=u.v$.
\end{definition}
Putting these conditions together yields the topological definition of CGD:
\begin{definition}[Causal Graph Dynamics]\label{def:causal}
A CGD is a shift-invariant, continuous, bounded dynamics.
\end{definition}

As a simple example we provide an original, general scheme for propagating particles on an arbitrary network in a reversible manner:
\begin{example}[General reversible advection]\label{ex:genadv} Consider $\pi=\{a,b,\ldots\}$ a finite set of ports, and let $\Sigma={\cal P}(\pi)$ be the set of internal states, where: $\varnothing$ means  `no particle is on that node'; $\{a\}$ means `one particle is set to propagate along port $a$'; $\{a,b\}$ means `one particle is set to propagate along port $a$ and another along port $b$'\ldots. Let $s$ be a bijection over the set of ports, standing for the successor direction. Fig. \ref{fig:genadv} specifies how individual particles propagate. Basically, when reaching its destination, the particle set to propagate along the successor of the port it came from. Missing edges behave like self-loops. Applying this to all particles synchronously specifies the graph dynamics. Because each port is used at most once, this dynamics is indeed well-defined. Moreover, it is indeed reversible: choosing $s'=s^{-1}$ lets the particles travel in the opposite direction.

\end{example}

\begin{figure}[h]
\begin{center}
\includegraphics[scale=2]{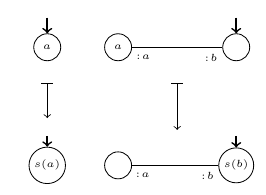}
\end{center}
\caption{\label{fig:genadv} {\em Generalized advection (Example \ref{ex:genadv})}. Diagram representing the two possible evolutions of a particle $a$, either the vertex contains no port $a$, in which case the particle remains motionless, or it moves along the port $a$. In both cases, the particle changes state. }
\end{figure}

Notice that advection is not only reversible, but time-symmetric. Indeed let $T$ be the dynamics which replaces, cell-wise, each particle $p$ by $s^{-1}(p)$.  Then $A^{-1} = TAT$. The physical analogous of $T$ is to multiply the momentum of each particle by $-1$.

\noindent {\bf Reversibility.} Invertibility is imposed in the most general and natural fashion.
\begin{definition}[Invertible dynamics]
A dynamics $(F,R_\bullet)$ is said to be invertible if $F$ is a bijection.
\end{definition}
\changes{
Unfortunately, this condition turns out to be very limiting. Indeed, except in a handful of cases, vertex creation/destruction is prohibited. The dynamics becomes essentially vertex preserving. 

\begin{example}[Hasslacher-Meyer: pointers break injectivity]
In Ex. \ref{ex:HM} the Hasslacher-Meyer dynamics was described informally, without keeping track of pointers, that is without specifying the $R_\bullet$ function. We do so in grey in Fig. \ref{fig:AHM}. Notice that case $(b)$ fails to be injective.    
\end{example}

Clearly, the root of this failure to be injective is the fact that if two vertices $u$ and $v$ merge into a vertex $w$ containing the pointer, there is ambiguity about the origin of the pointer. The only case where this ambiguity is not present is when $u$ and $v$ are symmetrical, as can be seen in Fig. \ref{fig:symetry}.

\begin{figure}
\begin{center}
\includegraphics[scale=1,clip=true]{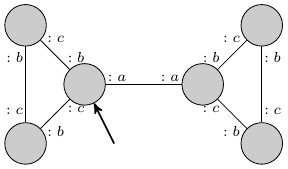}
\end{center}
\caption{\label{fig:symetry} {\em Example of a symmetrical pointed graph}. This graph is symmetrical because the $\varepsilon$ vertices are different vertices and $X=X_a$.}
\end{figure}

This limitation of reversible CGD has been captured in the following theorem---which the present paper seeks to circumvent: 
}
\begin{theorem}[Invertible implies almost-vertex-preserving]\cite{ArrighiRCGD}\label{th:preserv}\ \\
Let $(F,R_\bullet)$ be an invertible CGD. Then there exists a bound $p$, such that for any graph $X$, if $|V(X)| > p$ then $R_X$ is bijective.
\end{theorem}
\changes{Still, one should notice that although this theorem holds for every CGD over ${\cal X}_{\Sigma,\pi}$, it may sometimes fail to hold for partially-defined CGD, i.e. having domain ${\cal S}\subset{\cal X}_{\Sigma,\pi}$. Such graph subsets can, for instance, be defined by forbidding a given set of patterns to occur, cf. the notion of graph subshifts \cite{ArrighiSubshifts}. Throughout the paper we will make it clear which results could be threatened by a particular choice of ${\cal S}$.}

On the face of it reversibility is stronger a condition than invertibility:
\begin{definition}[Reversible Causal Graph Dynamics]
A CGD $(F,R_{\bullet})$ is {\em reversible} if there exists $S_\bullet$ such that $(F^{-1} ,S_{\bullet})$ is a CGD.
\end{definition}
Fortunately, invertibility gets you reversibility:
\begin{theorem}[Invertible implies reversible]\cite{ArrighiRCGD}\label{th:rev}
If $(F,R_\bullet)$ is an invertible CGD, then $(F,R_\bullet)$ is reversible.
\end{theorem}
\changes{(Notice that this theorem was only proven for CGD over ${\cal X}_{\Sigma,\pi}$. We conjecture that it still holds for CGD over graph subshifts.)}

\section{The invisible matter solution}\label{sec:IMCGD}
\subsection{Definition}
Reversible CGD are vertex-preserving. Still, we could think of using them to simulate a non-vertex-preserving dynamics by distinguishing `visible' and `invisible matter', and making sure that every visible node is equipped with its own reservoir of `invisible' nodes---in which it can tap. For this scheme to iterate, and for the created nodes to be able to create nodes themselves, it is convenient to shape the reservoirs as everywhere infinite binary trees.

\begin{definition}[Invisible matter graphs]\label{def:IMGraphs}\ \\
Consider ${\cal X}={\cal X}_{\Sigma,\pi}$, ${\cal T}={\cal X}_{\{\invstate\},\{m,l,r\}}$ and ${\cal X}'={\cal X}_{\Sigma\cup\{\invstate\},\pi\cup\{m,l,r\}}$, assuming that $\{\invstate\}\cap \Sigma=\emptyset$ and $\{m,l,r\}\cap \pi=\emptyset$. Let $T\in{\cal T}$ be the infinite binary tree whose origin $\varepsilon$ has a copy of $T$ at vertex $lm$, and another at vertex $rm$. Every $X\in {\cal X}$ can be identified to an element of ${\cal X}'$ obtained by 1/ attaching an instance of $T$ at each vertex through path $mm$, and 2/ closing under shift-invariance, thereby allowing that the pointer be located within one such $T$. The hereby obtained graphs will be denoted by ${\cal Y}$ and referred to as {\em invisible matter graphs}.
\end{definition}
For an example of an invisible matter graph, see Fig. \ref{fig:IMScattering}.

We will consider some CGD restricted to ${\cal Y}$, which we will call Invisible Matter CGD. Beforehand let us state two properties we want them to abide.\\
First, we want them trivial as soon as we dive deep enough into the invisible matter: 
\begin{definition}[Invisible matter quiescence]\label{def:IMQ}
A dynamic $(F,R_{\bullet})$ over ${\cal Y}$ is said {\em invisible matter quiescent} if there exists a bound $b$ such that, for all $Y \in {\cal Y}$, and for all $s,t$ in $\{lm,rm\}^*$, we have $|s|\geq b \implies R_{Y_{mms}}(t) = t $ and $\sigma_{F(Y_mms)}(t)=invisible$. 
%under shift-invariance, equivalent to initial definition (R_X(uv) = R_X(u)v)
\end{definition}
Essentially, if the pointer is deep enough in the invisible matter, then any tree of invisible matter above it is preserved (vertices preserve their paths and states).\\
Second, since vertex creation and destruction can be done by exhibiting or burying invisible matter, we restrict ourselves to dynamics that preserve all vertices:
\begin{definition}[Vertex preservation]\label{def:matterpreservation}
We say that a dynamic $(F,R_{\bullet})$ on ${\cal Y}$ {\em preserves vertices} if for all $Y \in {\cal Y}$, $R_Y$ is bijective.
\end{definition}

Notice that this condition is similar to boundedness, as it prevents nodes from splitting infinitely.

\begin{definition}[Invisible Matter Causal Graph Dynamics]\label{def:IMcausal}\ \\
A CGD over ${\cal Y}$ is said to be an IMCGD if and only if it is vertex-preserving and invisible matter quiescent.
\end{definition}

Fortunately, we are indeed able to encode non-vertex-preserving dynamics in the visible sector of an invertible IMCGD.

In order to illustrate this, we will take the basic but fascinating HM example \cite{MeyerLGA}. Again, this example features particles propagating, just as in Ex. \ref{ex:genadv}. However, whenever two particles meet the graphs grows or shrink, depending on how exactly they cross, as prescribed by a local permutation of two patterns. Strangely enough, in the long run, it typically just grows, thereby breaking the symmetry between past and future \cite{TimeArrow}. We have suspected, for quite a while, that the explanation for this strange behaviour was that the HM example was not truly reversible, in any rigorous sense. The following leads us to discard this explanation: here is a fully reversible implementation of the HM example, in the IMCGD formalism.

\begin{example}
    Let us consider $\mathcal{Y}$ the set of invisible matter graphs such that $\Sigma = \{\emptyset,a,b,ab,invisible\}$ and $\pi=\{a,b\} \cup \{m,l,r\}$. This can be seen as the set of graphs of degree at most two, with two types of particle $a$ and $b$ present on their respective ports. We can simulate the HM collision step by an IMCGD where the visible vertices behave in the following way :
    \begin{itemize}
        \item If vertex $v$ is such that $\sigma(v)=ab$, meaning that it holds both particles, then it simulates a split by lifting a vertex from the invisible matter, i.e. $R(v.mm)=R(v).ab$, whilst maintaining connections left  $R(v).ba = R(v.ba)$ and right $R(v.mm).ab = R(v.ab)$.
        The particles also separate according to $\sigma(R(v))=b$ and $\sigma(R(v).ab)=a$.
        The invisible matter is then restored by $R(v).mmt = R(v.mmlmt)$ and $R(v).abmmt=R(v.mmrmt)$ for all $t \in \{lm,rm\}^*$.
        \item If there is a pair of vertices $u,v$ such that $\sigma(v) = b$, $\sigma(u) = a$ and $v.ab = u$, then the merger of these vertices is implemented by following the inverse steps as the one described above.
        \item Other vertices are left unchanged.
    \end{itemize}
     Fig. \ref{fig:IMScattering} provides a more intuitive representation of how this dynamics operates.
\end{example}

\begin{figure}[h]
\begin{center}
\includegraphics[scale=1,clip=true,trim=0.5cm 0.5cm 2cm 0.4cm]{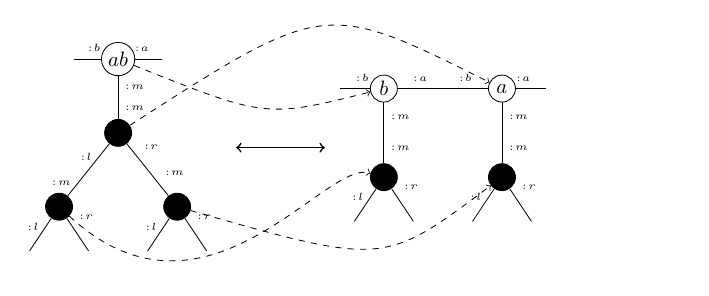}
\end{center}
\caption{\label{fig:IMScattering} {\em HM example's collision step with pointers and invisible matter.} Black vertices are `invisible'. The dotted lines show how vertices are mapped through this step. In particular, this tells where to place the pointer in the image.}
\end{figure}

\begin{example}[Invisible Matter HM]\label{ex:IMHM}
We alternate 1. the advection step of Example \ref{ex:genadv} with 2. the collision step shown in Fig. \ref{fig:IMScattering}. The composition of these two stages constitutes an invertible IMCGD.
%Consider {\cal X} as in Example \ref{ex:genadv} and extend it to {\cal Y}. Alternate: 1. a step of advection as in Example \ref{ex:genadv} and \ref{fig:AHM}$(a)$, 2. a step of collision,  where the collision is the specific graph replacement provided in Fig. \ref{fig:IMScattering}. The composition of these two specifies the invertible IMCGD.
\end{example}
%Notice how the graph replacement of Fig. \ref{fig:AHM}$(b)$---with the grey color taken into account---would fail to be invertible, due to the collapsing of two pointer positions into one.

\changes{
The following lemma is an indicator of the robustness of this definition: 
\begin{lemma}[Composability of IMCGD]{\label{lem:im-composability}}
  Let $(F,R_\bullet)$ and $(F',S_\bullet)$ be two IMCGD over $\mathcal{Y}$. Then $(F'\circ F,T_\bullet)$ with $T_Y=S_{F(Y)}\circ R_Y$ is also an IMCGD over $\mathcal{Y}$. 
\end{lemma}
\begin{proof}
In \cite{ArrighiCayleyNesme} it is proven that if $(F,R_\bullet)$ and $(F',S_\bullet)$ are CGD, then $(F'\circ F,T_\bullet)$ is also a CGD. Since both $S_\bullet$ and $R_\bullet$ are bijective, so is their composition and so $(F'\circ F,T_\bullet)$ is vertex-preserving.\\
It remains to prove that 
$(F'\circ F,T_\bullet)$, is invisible-matter quiescent. Let $b, b'$ be the bounds for $F$ and $F'$. We will show that for all $s'',t\in \{lm,rm\}^*$ with $|s''|\geq b''=b+b'$, we have  $T_{Ymms''}(t)=t$ and $\sigma_{F'(F(Y_{mms''}))}(t)=invisible$. We let $mms''=mm.s.s'$ with $|s|>b$ and $|s'|>b'$. As prior remarks we notice that $R_{Y_{mms''}}(t)=t$ and $R_{Y_{mms}}(s')=s'$ by the invisible matter quiescence of $F$. We also notice that $R_Y(mms)=mm\alpha$ with $\alpha\in \{lm,rm\}^*$, since by the invisible matter quiescence of $F$, $\sigma_{F(Y_{mms})}(\varepsilon)=invisible$ and $F(Y_{mms})\in \Y$. 
\begin{align*}
(S_{F(Y_{mms''})}\circ R_{Y_{mms''}})(t) &= S_{F(Y_{mms''})}(t) \quad\textrm{by the first remark}\\
&= S_{F(Y)_{R_Y(mms).R_{Y_{mms}}(s')}}(t) \quad\textrm{by shift-invariance}\\
&= S_{F(Y)_{mm\alpha.s'}}(t)\quad\textrm{by the above remarks}\\
&=t \quad\textrm{ by invisible-matter quiescence of $F'$}.
\end{align*}
\begin{align*}
    \sigma_{F'(F(Y_{mm.s.s'}))}(t)&=\sigma_{F'(F(Y)_{R_Y(mms).R_{Y_{mms}}(s')})}(t) \quad\textrm{by shift-invariance}\\
    &=\sigma_{F'(F(Y)_{mm\alpha.s'})}(t) \quad\textrm{by the above remarks}\\
    &=invisible \quad\textrm{ by invisible-matter quiescence of $F'$}.
\end{align*}
\end{proof}
}

\subsection{Compactness}

In order to prove that $(F^{-1},R_{F^{-1}(\bullet)}^{-1})$ is an IMCGD, we will need to prove that it is continuous, shift-invariant, vertex-preserving and invisible matter quiescent. Among these properties, only the continuity is difficult to prove. Intuitively this property over ${\cal Y}$ is inherited from that of CGD over ${\cal X}'$. Th. \ref{th:rev}, however, relies on the compactness of ${\cal X}$, and as a matter of fact ${\cal Y}$ is not compact. Still it admits a compact closure $\overline{{\cal Y}}$, over which IMCGD have a natural, continuous extension.

\begin{definition}[Closure]\label{def:closure}
The {\em compact closure} of ${\cal Y}$ in ${\cal X}'$, denoted $\overline{{\cal Y}}$, is the subset of elements $Y'$ of ${\cal X}'$ such that, for all $r$, there exists a $Y(r)$ in ${\cal Y}$ satisfying $Y(r)^r = Y'^r$.
\end{definition}
%\begin{definition}
%Let $Y'$ be  a graph in $\overline{{\cal Y}}$ with no point in the visible matter and let $(u_n)$ be the sequence of $\{lm,rm\}^*$ such that $Y^n = T_{u_n}^n$ and $u_n$ suffix of $u_{n+1}$. $Y$ being totally determined by the sequence $u$, we can write $Y = T_u$. The sequence $u$, growing for the suffix relation, can be identified with an infinite word of $\{lm,rm\}^{-\N}$, id est an infinite word with an end but no beginning.
%\end{definition}
In order to better understand this closure we establish the following preliminary result.

\begin{proposition}[Closure of visible]\label{prop:closure_visible_point}
Consider $Y'$ in $\overline{{\cal Y}}$ with $\varepsilon$ visible. Then $Y'$ is in ${\cal Y}$.
\end{proposition}

\begin{proof}
Consider $v$ visible in $Y'$. By the first part of Lemma \ref{lem:vispaths} the shortest path from $\varepsilon$ to $v$ is of the form $u\bdot t$ with $u$ in $\Pi^*$ and $t$ in $\{lm,rm\}^*$. But this $t$ needs to be the empty word, otherwise $v$ would be invisible. Therefore visible nodes form a $\Pi^*$-connected component, call it $X$. By the second part of Lemma \ref{lem:vispaths} each vertex of $X$ has, in $Y'$, an invisible matter tree attached to it---and no other invisible matter due again to the first part of Lemma \ref{lem:vispaths}. Finally, there is no other invisible matter in $Y'$ altogether, because $Y'$ is connected. \hfill$\square$ 
\end{proof}

\begin{proposition}[Closure characterization]\label{compact closure characterization}
	$$\overline{\cal Y} = {\cal Y} \cup \{T_u : u \in \{lm,rm\}^{-\N} \}$$ 
\end{proposition}

\begin{proof}
This is an immediate consequence of Lemmas \ref{lem:invpaths}, \ref{lem:vispaths}, \ref{lem:finiteinvroot} and Proposition \ref{prop:closure_visible_point}. Because these have simple but lengthy proofs, they are postponed until the \ref{subsec:IMCGDcompactness}. \hfill$\square$
\end{proof}

\noindent Now that we know what the closure of ${\cal Y}$ looks like, we can try to extend IMCGD to it.

\begin{theorem}[Compact extension of an IMCGD]\label{th:IMextension}
Consider $(F,R_\bullet)$ a continuous and shift-invariant dynamics over ${\cal Y}$. We have that $(F,R_\bullet)$ is invisible matter quiescent if and only if $(F,R_\bullet)$ can be continuously extended to ${\cal \overline{Y}}$ by letting $F(T_u) = T_u$ and $R_{T_u} = Id$ for any $u$ in $\{lm,rm\}^{-\N}$.
\end{theorem}
\newcounter{counterpropim}{\value{counterthm}}

\begin{proof}
This proof is quite technical, therefore the details of it have been pushed to \ref{subsec:IMCGDcompactness} (Proposition \ref{prop:IMextension_app}). The left to right implication is simply achieved through the continuity and invisible matter quiescence of $F$. Indeed, whatever the graph, if we dive deep enough into the invisible matter, it cannot be differentiated locally from $T_u$. Locally it has the same image as a tree of invisible matter. Using the invisible matter quiescence, we obtain that $F(T_u) = T_u$ and $R_{T_u} = Id$.
The right to left implication uses the fact that since ${\cal \overline{Y}}$ is compact, $(F,R_\bullet)$ is uniformly continuous by the Heine-Cantor Theorem. Using a similar reasoning as above, any graph is similar to $T_u$ if we dive deep enough into the invisible matter. Combined with the uniform continuity, this gives us the invisible matter quiescence of $(F,R_\bullet)$. \hfill$\square$
\end{proof}

Note that throughout the rest of the paper, IMCGDs will continue to have domain $\Y$, unless we make explicit use the above proposition to extend to $\overline{\Y}$.

\subsection{Reversibility}\label{subsec:reversibility_imcgd}
We are finally able to prove the reversibility Theorem : 

\begin{theorem}[Invertible implies reversible]\label{th:imrev}
If $(F,R_\bullet)$ is an invertible IMCGD, then $(F^{-1},R_{F^{-1}(\bullet)}^{-1})$ is an IMCGD, %with $R_{F^{-1}(\bullet)}^{-1}$ the function that maps $v'$ into $R_{F^{-1}(Y)}^{-1}(v')$. 
where $R_{F^{-1}(\bullet)}^{-1}$ is a shorthand notation for the family of functions $S$ such that for all $X$, $S_X = R_{F^{-1}(X)}^{-1}$.
\end{theorem}
\begin{proof}
Let $(F,R_\bullet)$ be an invertible IMCGD. Extend it to $\overline{\cal Y}$ as in Th.~\ref{th:IMextension}, and notice that it is still invertible. Let $F^{-1}$ be the inverse of $F$, and $S_\bullet$ be short for $R_{F^{-1}(\bullet)}^{-1}$.\\
$[$Shift-invariance$]$ Let $Y$ be in $\cal Y$ and $u$ be in $Y$. We have, by shift-invariance of $F$: 
$$F(F^{-1}(Y)_{S_Y(u)}) = F(F^{-1}(Y))_{R_{F^{-1}(Y)}(S_Y(u))} = Y_{R_{F^{-1}(Y)} \circ R_{F^{-1}(Y)}^{-1}(u)} = Y_u$$
Applying $F^{-1}$, on both sides we get $F^{-1}(Y)_{S_Y(u)} = F^{-1}(Y_u)$.\\
Let $Y$ be in $\cal Y$ and $u, v$ be in $Y$. On the one hand, we have: $$R_{F^{-1}(Y)}(S_Y(uv)) = R_{F^{-1}(Y)}(R_{F^{-1}(Y)}^{-1}(uv)) = uv$$
On the other hand, using the shift-invariance of $F$:
\begin{align*}
R_{F^{-1}(Y)}(S_Y(u)S_{Y_u}(v)) &= R_{F^{-1}(Y)}(S_Y(u))R_{F^{-1}(Y)_{S_Y(u)}}(S_{Y_u}(v))\\
	&= R_{F^{-1}(Y)}\circ R_{F^{-1}(Y)}^{-1} (u)\, R_{F^{-1}(Y_u)}\circ R_{F^{-1}(Y_u)}^{-1} (v) = uv
\end{align*}
Thus, $R_{F^{-1}(Y)}(S_Y(uv)) = R_{F^{-1}(Y)}(S_Y(u)S_{Y_u}(v))$. Applying $R^{-1}_{F^{-1}(Y)}$, we get $S_Y(uv) = S_Y(u)S_{Y_u}(v).$
Therefore $(F^{-1}, S_\bullet)$ is shift-invariant.\\
$[$Vertex-preservation$]$ The bijectivity of $S_Y$ for any $Y$ follows form the bijectivity of $R_Y$ for any $Y$. Boundedness follows from vertex-preservation.\\
$[$Continuity$]$ $F$ is an invertible continuous function over a compact space, so its inverse $F^{-1}$ is also continuous.\\
$[$Invisible matter quiescence$]$ We have extended $F$ to ${\cal \overline{Y}}$,  by letting $F_{|\overline{\cal Y} \setminus {\cal Y}} = Id$ and $R_{|\overline{\cal Y} \setminus {\cal Y}} = Id$, therefore we have $F^{-1}_{|\overline{\cal Y} \setminus {\cal Y}} = Id$ and $S_{|\overline{\cal Y} \setminus {\cal Y}} = Id$. By applying the converse of Th.~\ref{th:IMextension}, $F_{|{\cal Y}}^{-1}$ is invisible matter quiescent.\\
Altogether, we have proved that the inverse of $(F,R_\bullet)$ is also an IMCGD.
\hfill$\square$ \end{proof}

\section{The anonymous solution} \label{sec:ACGD}

Having a pointer is essential in order to express causality, but cumbersome when it comes to reversibility. A direct way to get the best of both worlds is to only consider the dynamics that we obtain by projecting causal dynamics over $\mathcal{X}_{\Sigma,\pi}$, onto $\mathcal{\anonymous{X}}_{\Sigma,\pi}$.
\begin{definition}[Anonymous Causal Graph Dynamics]
Consider $\anonymous{F}$ a function over $\mathcal{\anonymous{X}}_{\Sigma,\pi}$. We say that $\anonymous{F}$ is an ACGD if and only if there exists a CGD $(F,R_{\bullet})$ such that $F$ over $\mathcal{X}_{\Sigma,\pi}$ induces $\anonymous{F}$ over $\mathcal{\anonymous{X}}_{\Sigma,\pi}$. More formally, $\anonymous{F}$ is the dynamics such that for all $X \in \mathcal{X}$, $\anonymous{F(X)} = \anonymous{F}(\anonymous{X})$.
\end{definition}
\begin{remark}
The shift-invariance property, namely $F(X_u) = F(X)_{R_X(u)}$, entails that if two pointed graphs modulo differ just by their pointer, then so do their images. Dropping the pointer therefore unambiguously induces $\anonymous{F}$.
\end{remark}
\begin{proposition}
	For any pair of CGDs $F$ and $G$, we have : 
	    $$\anonymous{(F \circ G)} = \anonymous{F} \circ \anonymous{G}$$
\end{proposition}

\begin{proof}
Consider $X\in \X$. By definition of induced dynamics, we have : 
$$\anonymous{(F \circ G)}(\anonymous{X})=\anonymous{(F(G(X)))}=\anonymous{F}(\anonymous{(G(X))})=(\anonymous{F} \circ \anonymous{G}) (\anonymous{X})$$\hfill$\square$ 
 \end{proof} 
%\begin{proposition}
%	$$\anonymous{F} \circ \anonymous{G} = \anonymous{F \circ G}$$
%\end{proposition}
%\begin{proof} $G(X)$ can be chosen to represent $\anonymous{G}(\anonymous{X}) = \anonymous{G(X)}$, so $F(G(X))$ represents both $\anonymous{F} \circ \anonymous{G} (\anonymous{X})$ and $ \anonymous{F \circ G} (\anonymous{X})$, which are then equals.
%\hfill$\square$ \end{proof}
Invertibility is again imposed in the most general and natural fashion :
\begin{definition}[Invertible]
An ACGD $\anonymous{F}$ is said to be {\em invertible} if $\anonymous{F}$ is bijective.
\end{definition}
\VC{Invertibility, then, just means that $\anonymous{F}$ is bijective. } Fortunately, this time the condition is not so limiting without a pointer, and we are able to implement non-vertex-preserving dynamics, as can be seen from this slight generalization of the HM example:

\begin{figure}
\begin{center}
\includegraphics[scale=1,clip=true]{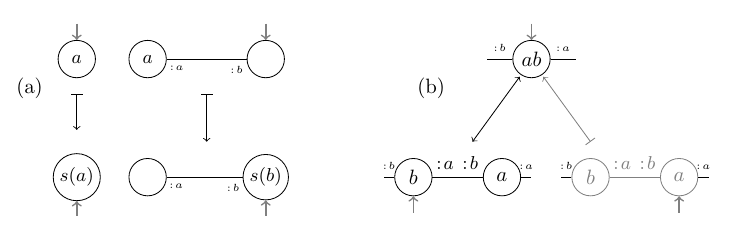}
\end{center}
\caption{\label{fig:AHM}{\em $(a)$ Advection step. $(b)$ Collision step.}}
\end{figure}
%\begin{figure}[h]
%\begin{center}
%\includegraphics[scale=1,clip=true]{Meyer_scattering.pdf}
%\end{center}
%\caption{\label{fig:AHM} {\em $(a)$ General reversible advection. $(b)$ The Hasslacher-Meyer example's collision step.} The anonymous dynamics is in plain black, the underlying regular dynamics is in grey.}
%\end{figure}

\begin{example}[Anonymous Hasslacher-Meyer]\label{ex:AHM} Once again we alternate between an advection step and a collision step, as can be seen in Fig. \ref{fig:AHM}. The anonymous dynamics is shown in black, whilst the underlying pointed dynamics are shown in grey. 
\end{example}
%\begin{example}[Anonymous HM]\label{ex:AHM}Consider the state space of Example \ref{ex:genadv} and alternate: 1. a step of advection as in Fig. \ref{fig:AHM}$(a)$, 2. a step of collision, where the collision is the specific graph replacement provided in Fig. \ref{fig:AHM}$(b)$. The composition of these two specifies the ACGD.
%\end{example}

Notice that the original HM example features just particles moving on a circle in one direction or the other, stretching or contracting space according to how they meet. This corresponds to restricting to $\pi = \{a,b\}$ and $ab-$edges only, which is still a compact space \cite{ArrighiCGD}.

One can ask if such model leads to fundamentally different dynamics than IMCGD. We answer in the negative, and by doing so we prove that invertibility implies reversibility. However, this is a quite technical result, therefore the formal proofs have been kept in the next section.
%\noindent In what follows $\alpha$ is the natural, surjective map from ${\cal Y}$ to ${\cal \anonymous{X}}$, which (informally): 1. Drops the pointer and 2. Cuts out the invisible matter. Whatever an ACGD does to a $\alpha(Y)$, an IMCGD can do to $Y$---moreover the notions of invertibility match:

%----------------------------------------------

\subsection{Simulation of the invisible by the anonymous}

First of all we need an operation which brings the pointer back into the visible matter. 
\begin{definition}[Inversion and projection of paths]
For any graph $Y\in \Y$, and $u\in V(Y)$, let $\projection{u} \in ((\pi\setminus{\{m,l,r\})^2})^*$ be the visible part of $u$ i.e. the only path such that there exists $s,t \in (\varepsilon|mm).\{lm,rm\}^*$ with $u = \reversepath{s}\projection{u}t$. 
%This vertex is unique because every tree of invisible matter has only one root in visible matter. 
We denote by $\projection{Y}$ the \emph{projection} of $Y$, i.e. the graph obtained by projecting the pointer into the visible matter. We denote by $\projection{\Y}$ the image of this function. Clearly, there is a bijection between $\projection{\Y}$ and $\X$, obtained by simply by discarding/reintroducing the invisible matter.

\end{definition}

\begin{remark}
For any $Y \in \Y$ and $u \in V(Y)$, $\projection{u}$ is in $V(\projection{Y})$. Moreover, $Y \in \projection{\Y}$ if and only if $Y=\projection{Y}$.
\end{remark}

\begin{lemma}\label{lem:projdistrib}
  Let $Y \in\Y$. For all $u,v \in (\pi \cap \{m,l,r\})^{*}$ such that $uv \in V(Y)$, we have that $\projection{(uv)}= \projection{u} \projection{v}$.
\end{lemma}

\begin{proof}
	Let $Y\in\Y$. For all $u,v \in (\pi \cap \{m,l,r\})^{*}$ such that $uv \in V(Y)$ let $s,r,r',t \in (\varepsilon|mm).\{lm,rm\}^*$ such that $u = \reversepath{s}\projection{u}r'$ and $v=\reversepath{r}\projection{v}t$. Since $uv$ is a valid path of $Y$ and the uniqueness of the projection, we have that $r=r'$. This allows us to conclude with the following equalities: 
	$$uv = \reversepath{s}\projection{u}r\reversepath{r}\projection{v}t = \reversepath{s}\projection{u}\projection{v}t = \reversepath{s}\projection{uv}t$$ \hfill$\square$
\end{proof}

\begin{definition}[Projection of a dynamics]\label{def:proj}
Let $(F,R_{\bullet})$ be a dynamics on ${\cal Y}$. We define the dynamics $(F^{\lrcorner},R^{\lrcorner}_{\bullet})$ on $\projection{\Y}\cong \X$ by :
	\begin{itemize}
		\item[$\bullet$] $F^{\lrcorner}(Y)=\projection{(F(Y))}$ 
		\item[$\bullet$] $R^{\lrcorner}_Y(u) = \projection{(R_Y(u))}$.
	\end{itemize}
for all $Y \in \projection{\Y}$ and for all $u \in \Pi^*$. We refer to $(F^{\lrcorner},R^{\lrcorner}_{\bullet})$ as the {\em projection} of $(F,R_{\bullet})$. 
\end{definition}

\begin{proposition}\label{proj_shift_inv}
	The projection of a translation invariant dynamics is translation invariant. 
\end{proposition}

\begin{proof}
	Let $(F,R_{\bullet})$ be a translation-invariant dynamic on ${\cal Y}$. Let $Y \in \projection{\Y}$ and $u$ be a vertex of $Y$. By definition of $F^{\lrcorner}(Y)$ and by translational invariance of $(F,R_{\bullet})$ we have : $$F^{\lrcorner}(Y_u)=\projection{F(Y_u)}=\projection{(F(Y)_{R_Y(u)})}= F^{\lrcorner}(Y)_{\projection{(R_Y(u))}}= F^{\lrcorner}(Y)_{R^{\lrcorner}_Y(u)}$$
 
	Since $(F,R_{\bullet})$ is invariant by translation, and $\projection{(uv)}= \projection{u} \projection{v}$ we have the following equalities:
	$$R^{\lrcorner}_Y(uv) = \projection{(R_Y(uv))}= (R_Y(u)R_{Y_u}(v))^{\lrcorner} = \projection{R_Y(u)} \projection{R_{Y_u}(v)}= R^{\lrcorner}_Y(u)R^{\lrcorner}_{Y_u}(v)$$
	$(F^{\lrcorner},R^{\lrcorner}_{\bullet})$ is therefore invariant by translation. \hfill$\square$
\end{proof}

Before proving the continuity of the projection, we need to prove that the pointer does not plunge too deeply into the invisible matter in a single time step. This will allow us to prove that the projected dynamics merges a bounded number of vertices, which will be useful for proving continuity.

\newcommand{\projectedY}{\projection{{\cal Y}}}

	\begin{lemma}\label{coquiescence}
		Let $(F,R_{\bullet})$ be an IMCGD. There exists a bound $b\in \setN$ such that, for all $Y,Z \in\projectedY$, for all $s\in\IM$, $F(Y) = Z_{mms}$ implies $|s| \leq b$.
	\end{lemma}
	
	\begin{proof}
		By contradiction, let $(F,R_{\bullet})$ be an IMCGD, and suppose that for all $k \in \projectedY$ there exists $Y_k, Z_k \in \projectedY$, and $s_k\in\IM$ such that $F(Y_k) = (Z_k)_{mms_k}$ with $|s_k| > k$.
		By extending $(F,R_{\bullet})$ to be over $\overline{\Y}$ by Th.~\ref{th:IMextension} and using compactness, $(Y_k)_{k\in \setN}$ admits a convergent sub-sequence in ${\overline{\cal Y}}$. Let $(Y'_k)_{k\in \setN}$ be one of these sub-sequences and $Y'$ its limit.
		For all $k$, $Y'_k$ has its pointer in the visible matter, and so $Y' \in \projectedY$ by Proposition~\ref{prop:closure_visible_point}. Let $Z \in \projection{{\cal Y}}$ and $s \in \{lm,rm\}^*$ such that $F(Y') = Z_{mms}$. For all $k > |s|$, $(F(Y'_k))^{|s|}$ has no visible matter because $|s_k|>|s|$, so this is necessarily the case for $(F(Y'))^{|s|}$, which contradicts $F(Y') = Z_{mms}$. \hfill$\square$
	\end{proof}

	\begin{proposition}\label{proj_continuous}
		The projection of an IMCGD is continuous.
	\end{proposition}
	
	\begin{proof}
	Let $(F,R_{\bullet})$ be an IMCGD. Let $b$ be the bound given by Lemma~\ref{coquiescence}.
	
	Let $Y \in \projectedY$ and $m \in \setN$. By continuity of $(F,R_{\bullet})$, there exists $n\geq0$ such that for all $Z \in {\cal Y}$, $Y^n = Z^n$ implies :
	\begin{itemize}
		\item[(a)] $(F(Y))^{m+2b} = (F(Z))^{m+2b}$.
		\item[(b)] $\dom\,{R_Y}^{m+2b}\subseteq V(Y^n)$, $\dom\, {R_Z}^{m+2b}\subseteq V(Z^n)$, and ${R_Y}^{m+2b} = {R_Z}^{m+2b}$.
	\end{itemize}
	
	%F\proj(Y)^m={F\proj(Z^n)}^m
	In particular, for $Y \in \projectedY$, if we denote $F(Y) = F^{\lrcorner}(Y)_{.s}$ with $s\in\IM$, we obtain by translating by $s$ on both sides and applying $(a)$:
		$$(F^{\lrcorner}(Y))^{m+b} = (F^{\lrcorner}(Z))^{m+b}$$
	a fortiori
	$$(F^{\lrcorner}(Y))^m = (F^{\lrcorner}(Z))^m$$
	
	Let $v \in (F^{\lrcorner}(Y))^m$. Let $u$ be an antecedent of $v$ by $R^{\lrcorner}_Y$. There exists $s,t\in (\varepsilon|mm).\{lm,rm\}^*$ such that $R_Y(u) = svt \in (F(Y))^{m + 2b}$. By $(b)$ we have $u \in \dom\,{R_Y}^{m+2b} \subseteq V(Y^n)$ and
	$$\dom\,{R^{\lrcorner}_Y}^m\subseteq \dom\,{R_Y}^{m+2b} \subseteq V(Y^n)$$
	
	Similarly,
$$\dom\,{R^{\lrcorner}_Z}^m\subseteq \dom\,{R_Z}^{m+2b} \subseteq V(Z^n)$$
	
	Let $u$ belongs to $\dom\,{R^{\lrcorner}_Y}^m$. 
 $u$ also belongs to $\dom\,{R_Y}^{m+2b}$, and so
$$ \reversepath{s}.R^{\lrcorner}_Y(u).t = R_Y(u) = R_Y^{m+2b}(u) = R_Z^{m+2b}(u) = {\reversepath{s}}.R^{\lrcorner}_Z(u).t $$
	
	So finally we have $${R^{\lrcorner}_Y}^m = {R^{\lrcorner}_Z}^m $$ which concludes the proof of continuity of $(F^{\lrcorner}, R_{\bullet})$.  
	\hfill$\square$
	\end{proof}
	
	We will now concentrate on proving that the projection is bounded.
	
\begin{lemma}[Scattering bound]\label{BS lemma}
	For any CGD $(F,R_{\bullet})$, there exists a bound $c$ such that for any $u \in V(X)$ $|u| = 2 \implies |R_X(u)| \leq c$.
\end{lemma}

Intuitively, by continuity it is sufficient to look at the set of disks with radius $2$ and take $c$ to be the maximum radius produced by the dynamics on these disks.

\begin{proof}
	By contradiction, let $(F,R_{\bullet})$ be a dynamics and suppose that for all $c$, there exists a graph $X_c$ and a vertex $u_c \in V(X)$ such that $|u_c| = 2$ (adjacent vertex, at a distance of $2$ ports), and $R_X(u_c) > c$. Since the set $\pi$ of ports is finite and using the pigeonhole principle, we can extract a sequence $(X_c)$ in such a way as to keep $u$ constant. By compactness of ${\cal X}$, we can also require this sequence to be convergent and denote its limit by $X$ (a fortiori, $(X_c)_{c\in \setN}$ always satisfies $\forall c\geq 0, |R_{X_c}(u)| > c$). 
	
	Using the continuity of $(F,R_{\bullet})$, there exists an $n$ such that, for any graph $S$, $X^n = S^n$ implies $R_X^{|R_X(u)|} = R_S^{|R_X(u)|}$. Let $p\in\setN$ be such that $X_p^n = X^n$ and such that $|R_{X_p}(u)| > |R_X(u)|$ (i.e. $p >|R_X(u)|$). This leads to a contradiction because $R_X(u) = R_X^{|R_X(u)|}(u) = R_{X_p}^{|R_X(u)|}(u) = R_{X_p}(u)$. \hfill$\square$
\end{proof}

Since we used the compactness of $\X$, this lemma does not apply to all causal dynamics of partial graphs. However, it does apply to IMCGDs, via the application of Proposition \ref{compact closure characterization}.

\begin{corollary}\label{BS corollary}
	For any CGD $(F,R_{\bullet})$, there exists a bound $s$ such that $|R_X(u)| \leq s \times |u|$.
\end{corollary}

\begin{proposition}
 The projection of an IMCGD is bounded.
\end{proposition}

This bound is a direct consequence of the quiescence of invisible matter, since it imposes a bound on the entry of invisible matter into visible matter.

\begin{proof}
We must prove that there exists a bound $b$ such that for all $Y\in\projection{\Y}$ and for all $w'\in \projection{F}(Y)$, 
there exists 
$u'\in \textrm{Im}(\projection{R}_Y)$ 
and $v'\in \projection{F}(Y)_{u'}^b$ 
such that $w'=u'.v'$.
Let $(F,R_{\bullet})$ be an IMCGD. Let $Y \in \projection{\Y}$ and $y \in F^{\lrcorner}(Y)$ such that $y$ is not in the image of $R^\lrcorner_Y$.
$y$ also belongs to $V(F(Y))$, and by preservation of the vertices, it has an antecedent $us$ such that $R_Y(us)=y$. By quiescence of invisible matter, $s$ is bounded. And by Corollary \ref{BS corollary}, this implies that the distance between $R_Y(u)$ and $R_Y(us)=y$ is bounded. Since $R_Y(u) = t^{-1}R^{\lrcorner}_Y(u)t'$ for some $t$ and $t'$ bounded by Lemma~\ref{coquiescence}, the distance between $y$ and an element of $R^{\lrcorner}_Y(u)$ in $\textrm{Im}(R^{\lrcorner}_Y)$ is bounded. \hfill$\square$
\end{proof}

By combining the three previous properties, we finally obtain the simulation theorem :
\begin{theorem}\label{th:imcgd_to_cgd}
For any IMCGD $(F,R_{\bullet})$, its projection $(F^{\lrcorner},R^{\lrcorner}_{\bullet})$ is a CGD.
\end{theorem}

This CGD can easily be transformed into an anonymous dynamics by forgetting the pointer. 

\begin{theorem}\label{thm:ir_to_ar}
	Let $(F,R_{\bullet})$ be an IMCGD. $(F,R_{\bullet})$ is invertible if and only if $\anonymous{(F^\lrcorner)}$ is invertible. 
\end{theorem}

\begin{proof}
	Let $(F,R_{\bullet})$ be an IMCGD. Let $(F,R_{\bullet})$ be an invertible and therefore reversible IMCGD. For all $X \in \mathcal{X}$:
	
$$F^\lrcorner(X)\sim F(X)$$
	
	and so
	$$\anonymous{(F^\lrcorner)} = \anonymous{F}$$
	And finally
	$$\anonymous{F} \circ \anonymous{(F^{-1\lrcorner})} = \anonymous{F} \circ \anonymous{(F^{-1})} = \anonymous{(F \circ F^{-1})} = Id$$
	Similarly, we have $\anonymous{((F^{-1})^\lrcorner)} \circ \anonymous{F} = Id$, so $\anonymous{F}$ is reversible.
	
	Now suppose that $\anonymous{F}$ is invertible. We want to construct an IMCGD $(H,S_{\bullet})$ inverse of $(F,R_{\bullet})$. Let $Y\in\mathcal{Y}$ and $X\in \X$, such that $X$ is the graph obtain by removing the invisible matter from $\projection{Y}$. Let $u$ be a vertex of the anonymous graph $(\anonymous{F})^{-1}(X)$, consider the pointed graph $((\anonymous{F})^{-1}(X),u)$. We have :
	
$$F(((\anonymous{F})^{-1}(X),u)) \sim F^\lrcorner(((\anonymous{F})^{-1}(X),u)) \sim X$$
	$F(((\anonymous{F})^{-1}(X),u))$ is therefore equal to $X_{vt}$ for some v and a certain $t \in (\varepsilon|mm).\IM$. By bijectivity of $R_{(\anonymous{F})^{-1}(X)_u}$, there exists $w$ such that $F((\anonymous{F})^{-1}(X)_{uw}) = X$ (more precisely, $w$ is the antecedent of $vt$). We can define $H(Y)$ to be the function such that $(\anonymous{F})^{-1}(X)_{uw}$, and $S_Y$ such that $(R_{(\anonymous{F})^{-1}(X)_{uw}})^{-1}$, which proves that $(F,R_{\bullet})$ is invertible. 
    \hfill$\square$
    \end{proof}

\subsection{Simulation of the anonymous by the invisible}
\newcommand{\ports}{\pi}

\changes{(Notice that the results of this subsection may fail to hold for some partially-defined CGD, i.e. having domain over certain subsets ${\cal S}\subseteq{\cal X}_{\Sigma,\pi}$.)}

As stated in the introduction, vertex creation and destruction can be simulated by moving vertices to and from invisible matter. However, when simulating an ACGD, moving a vertex into invisible matter can break determinism or translation invariance, as shown in Example \ref{fig:turtle}. In this section, we will see that this can only happen a finite number of times. More precisely, it occurs when a graph has a translational symmetry that can be detected by the continuity radius.

\begin{figure}[h]
\begin{center}
\includegraphics[scale=1]{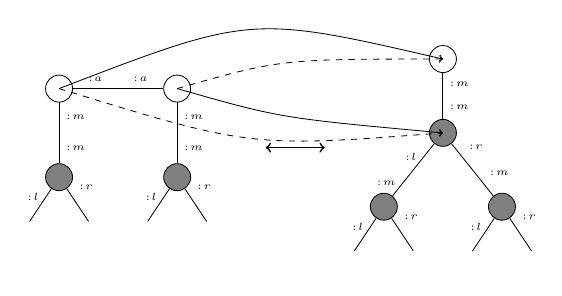}
\end{center}

\caption{\label{fig:turtle} {\em {The turtle and invisible-matter graph dynamics}. In the turtle dynamics, nothing happens except a permutation between the one-vertex graph and the two-vertex graph connected by an edge labeled $\{:a,:a\}$. 
this is a perfectly valid CGD, but it is not the projection of any IMGCD. Since the two vertices are symmetrical, we cannot decide which one should be sent into the invisible matter without breaking the translational invariance. The solid and dotted arrows show the two possible choices for vertex evolution.}}
\end{figure}

A symmetric graph is formally defined as follows:

\begin{definition}[Symmetric graph]
Given a graph $X\in{\cal X}_{\Sigma,\ports}$, we say that $X$ is symmetric if and only if there exists $u \in V(X)$ such that $u\neq \varepsilon$ and $X=X_u$. 
\end{definition}

We rely on this already known lemma:

\begin{lemma}[Structure of symmetric graphs]\cite{ArrighiRC}
Let $X\in{\cal X}_{\Sigma,\ports}$ be a symmetric graph. Then we have :
$X=\bigcup_{u\in T}u.G$ where
\begin{itemize}
\item $T$ is a vertex-transitive graph. 
\item $V\subseteq V(X)$ such that $\varepsilon\in V$ and $\{u.V\}_{u\in T}$ form a partition of $V(X)$.
\item $G=G(X)^0_{V}$ with $G(X)$ as defined in \ref{def:associatedgraph}. 
\item $w\sim w'$ if and only if $w=u.v$, $w'=u'.v$, $u,u'\in T$, and $v\in V(G)$.
\end{itemize}
\end{lemma}

From this result, we can prove the existence of a bound on the size of graphs performing symmetric fusion.

\begin{lemma}[Bound on symmetric merges]\label{lem:borned_symmetrical_fusion}
Let $(F,R_{\bullet})$ be a causal graph dynamics. There exists a bound $d$ such that, for all $X \in {\cal X}$, and $u \in V(X)$ such that $u \neq \varepsilon$, ($R_X(u) = R_X(\varepsilon) $ and $ X_u = X$) $\implies X = X^d$. 
\end{lemma}
The idea of this proof is that if two vertices are symmetric, then their neighbours share this symmetry. By merging them, by translation invariance and the fact that each port is used at most once, their neighbors must also merge. Reasoning by induction, this fusion propagates to the rest of the graph and, if the graph is large enough, outside the causal radius.  We can then modify the graph in such a way as to ensure that two very distant vertices merge, thus breaking the continuity.

\begin{example}[Symmetric ladders do not merge]
Let $\pi=\{a,b,c\}$ and $\Sigma=\varnothing$. Consider $X$ a very long ladder of made of two parallel $ab$-lines of vertices, connected by $cc$-rungs, and pointed towards the middle. Since $X_{cc}=X$, can $F(X)$ merge the ladder uprights into a single $ab$-line? Suppose that it can, and consider $Z$ with many top rungs missing, and with one of the upright made shorter. By locality and translation invariance, $F$ needs treat the bottom of $Z$ in the same way it treated $X$, i.e. by merging. However the top of the uprights of $Z$ are loose, arbitrarily far apart, and cannot be merged. Thus $F(Z)$ will be made of a single line at the bottom, and two loose lines on top. Consider the node at the fork. Each port is used only once, so one loose line needs to attach to port $c$, whilst the other is attached to $b$. But remember that one of the uprights was shorter, leading to a shorter image. Which should be attached to $c$, the shorter, or the longer one? Locally at the level of the fork, $Z$ still looks symmetric by $cc$ and so it is impossible for $F$ to decide. Thus the $X$ ladder cannot be merged in the first place.
\end{example}

\begin{proof}
 Let $(F,R_{\bullet})$ be a causal graph dynamics on ${\cal X}_{\Sigma,\ports}$. Let $n_0$ be the uniform continuity bound of $F$ for $m=0$, i.e. for all $X \in \setG$, we have $(F(X))^0=(F(X^{n_0}))^0$, $\dom\,R_{X}^0\subseteq V(X^{n_0})$ and $R_{X}^0=R_{X^{n_0}}^0$. By contradiction, suppose there exists a graph $X$ and a vertex $u\in V(X)$ such that $R_X(u) = R_X(\varepsilon)$, $X_u = X$ and $X^{2n_0} \neq X$. 
 
Since $X^{2n_0} \neq X$, there exists a vertex $v$ such that $v\notin V(X^{2n_0})$ and $v$ is a simple path (without loops) from $\varepsilon$. Let $v'$ be the shortest prefix of $v$ such that $v' \notin V(X^{n_0})$.
 We have $|v|> 2n_0$ and $|v'|=|n_0|$, so $|v'^{-1}v| > n_0$. There is a vertex $w \in V(X)$ which is a prefix $v$, such that $w \notin V(X^{n_0})$ and $u.w \notin V(X^{n_0})$.
 
Let $Z$ be the graph obtained by performing the following operations:

\begin{itemize}
    \item Cut all the edges $xy$ such that $x \notin V(X^{n_0})$ and $y \notin V(X^{n_0})$ except the ones on the paths $w$ and $uw$.
    \item Delete all vertices not belonging to the same connected component as the pointer.
    \item $w$ and $uw$ therefore have only one connected port, named $p_1$ here, so they have at least one free port $p_2$. We extend $w$ and $uw$ so that $w(p_2p_1)^{n_0} \in V(Z)$ and $uw(p_2p_1)^{n_0}\in V(Z)$.
\end{itemize}
We obtain a graph $Z$ that is locally similar to $X$, and has at least two arbitrarily removed paths. Since $R_X(u)=R_X(\epsilon) \in V(X^{n_0})$ and  $X^{n_0}=Z^{n_0}$, we obtain the following equalities: $R_Z(u) = R_Z^{n_0}(u) = R_X^{n_0}(u) = R_X^{n_0}(\varepsilon) =R_Z^{n_0}(\varepsilon)=R_Z(\varepsilon)$.  
Using translation invariance, we obtain that $R_Z(w(p_2p_1)^{n_0})=R_Z(uw(p_2p_1)^{n_0})\in V(Z)$ which contradicts the continuity of $(F,R_\bullet)$ as $w(p_2p_1)^{n_0}$ and $uw(p_2p_1)^{n_0}$ are at a distance greater than $n_0$. \hfill$\square$
 \end{proof}

 Once again, we'll use a construction from \cite{ArrighiRC} to break the symmetry of a graph.
 
\begin{definition}[Asymmetric Extension]\label{def:asymext}\cite{ArrighiRC}
	Given a finite symmetric graph $X\in{\cal X}_{\Sigma,\pi}$, we obtain an {\em asymmetric extension} $\assymext{X}$ by performing one of the following operations:
	\begin{itemize}
		\item Choose a vertex $w\in V(X)$ with a free port $p_1$ and connect a new vertex $w.p_1p_2$ to it, on its $p_2$ port.
		\item Choose $w\in V(X)$ which is part of a cycle, delete an edge $e$ from this cycle $w$ connecting $w$ and $w'$, then add two vertices $w.e$ and $w'.e^{-1}$.
	\end{itemize}
\end{definition}

\begin{lemma}[Asymmetry of the asymmetric extension]\cite{ArrighiRC}\label{lem:assymetricext}.
	Given a finite symmetric graph $X\in{\cal X}_{\Sigma,\pi}$, its asymmetric extension $\assymext{X}$ is asymmetric, and $|\assymext{X}| \leq |X| + 2$.
\end{lemma}

We have proved that only symmetric graphs with radii smaller than the causal radius can achieve symmetric fusion. This impossibility naturally extends to asymmetric graphs, when the asymmetry is not locally detectable.

\begin{lemma}[Local asymmetry]\label{lem:local_asymmetry2}
	Given a CGD $(F,R_{\bullet})$, we have that for almost any graph, vertices merging together by applying $F$ are asymmetric. More formally, there exists a bound $d$ , such that for any graph $X$ of radius greater than $d$ ($X \neq X^d$), and for any pair of vertices $v,v' \in V(X)$, $R_X(v) = R_X(v')$ and $v \neq v'$ implies $X_v \neq X_{v'}$.
\end{lemma}

%* page 27 Lemma 7, line 6 of the proof: first using "for all u" here is
%confusing since u already designates a particular vertex; second, I
%don't understand the conclusion "R_Z(u)=ε", did you forget the
%assumption "R_X(u)=ε"? If yes, it has to be justified.

\begin{proof}
	By contradiction, suppose that for all $d$, there exists $X\in \setG $ a graph of radius greater than $d$, and $v,v' \in X$ such that $R_X(v) = R_X(v')$, $v \neq v'$ and $X_v = X_{v'}$.
	
	Assume without loss of generality that $v = \varepsilon$. Let $n$ be the uniform continuity bound of $(F,R_{\bullet})$, for $m=0$. In particular we have that for all $Z\in{\cal X}$ and for all $u \in V(Z)$ such that $|u|\leq n$, $Z^n=X^n \implies R_Z(u) = \varepsilon$. Let $\assymext{X}$ be an asymmetric extension of $X$ obtained by considering one of the farthest vertices of $\varepsilon$. $\assymext{X}$ is therefore asymmetric and also verifies $(\assymext{X})^d = X^d$. As with $d \geq n$, we have can apply the uniform continuity and obtain that $R_{\assymext{X}}(v') = \varepsilon$. Due to the asymmetry of $\assymext{X}$, $\assymext{X} \neq \assymext{X}_u$ and the three previous assertions contradict Lem.~\ref{lem:borned_symmetrical_fusion}.\hfill$\square$
\end{proof}

Now that we have proved that CGDs, and therefore ACGDs, only perform a finite number of symmetric mergers we can finally state the simulation theorem :

\begin{theorem}[Extension of a CGD into an IMCGD]\label{th:cgd_to_imcgd}\ \\
For any CGD $(F,R_{\bullet})$, there exists an IMCGD $(F',R'_{\bullet})$ whose projection $(F'^{\lrcorner},R'^{\lrcorner}_{\bullet})$ is isomorphic to $(F,R_{\bullet})$ for any graph in $X\setminus S$, where $S$ is a finite set of symmetric graphs.
\end{theorem}

%⁃ notation S_{1Y} is confusing, better use T_Y

\begin{proof}
	Let a CGD $(F,R_{\bullet})$ over $\mathcal{X}$, and an isomorph CGD $(F'^{\lrcorner},R'^{\lrcorner}_{\bullet})$ over $\projection{\mathcal{Y}}$ obtained by adding invisible matter trees to very vertices. To extend it into an IMCGD, we need for any graph $Y \in {\cal Y}$ to make explicit its visible vertices, as well as the provenance of the visible vertices of $\projection{F'}(Y)$, so as to restore the injectivity and surjectivity of $\projection{R'}$ respectively. Of course, we can do this using invisible matter trees, but we must be careful not to break the other properties of $\projection{F'}$. \\
	
	By continuity of $\projection{R'}$, only a finite number of vertices can merge simultaneously at the same point ($dom((\projection{R'}_Y)^0) \subseteq V(Y^n)$, for some $n$). Since $\projection{R'}$ is bounded, we are guaranteed to find enough invisible matter for each of the vertices with no antecedent in $\projection{R'}_Y$. Combining these two properties, it is therefore possible to restore vertex preservation by IMCGD by moving only a bounded number of vertices from or to the invisible matter, i.e. without breaking the quiescence of the invisible matter.\\
	
	{\em Restoring injectivity : }
	First, we will restore injectivity by defining a function $S_{Y}$, such that for all $x_1,x_2 \in V(Y)$, $S_{Y}(x_1)$ and $S_{Y}(x_2)$ share the same root in invisible matter if and only if $\projection{R'}(x_1) = \projection{R'}(x_2)$. Let $y$ be a vertex of $Im(\projection{R'}_Y)$, i.e. a root in visible matter with an antecedent in $\projection{R'}_Y$. All we need to do is define a total, local, translation-invariant order over $(\projection{R'}_Y)^{-1}(y)$. We can then pose $(\projection{R'}_Y)^{-1}(y) = \{x_0, x_1, ...\}$, and using the continuity of $\projection{R'}$ show that $(\projection{R'}_Y)^{-1}(y)$ is finite. We can easily conclude by defining $S_{Y}(x_0) = y$, $S_{Y}(x_{1 + i}) = y.m.l^i$, $S_{X}(x_{1 + i}.m) = y.m.l^i.r$ and $S_{Y}(x_0.m) = y.m.l^{|(\projection{R'}_Y)^{-1}(y)| - 1}$. To order $(\projection{R'}_Y)^{-1}(y)$, one can define an order on $\{Y_x| \projection{R'}_Y(x) = y\}$, i.e. order them according to their translation symmetry class (in fact, nothing finer can be done without violating translation invariance). Lemma \ref{lem:borned_symmetrical_fusion} (local asymmetry) establishes exactly that this ordering is local. Lemma \ref{lem:local_asymmetry2} (second asymmetry) states that $ x \mapsto Y_x $ is injective, proving that an order on $\{Y_x| \projection{R'}_Y(x) = y\}$ actually provides an order on $(\projection{R'}_Y)^{-1}(y)$. Both lemmas apply only for $Y$ large enough.\\

{\em Restoring surjectivity : }
Let $u$ be a vertex of $V(\projection{F'}(Y))\backslash Im(\projection{R'}_Y)$. Let $u'$ be the vertex closest to $u$ in $Im(\projection{R'}_Y)$, taking the lexicographic order of the paths to break the equalities. This choice is relative to $u$, and therefore independent of where $Y$ is pointed, or to put it briefly, translation invariant. It is necessary for $u$ to come from the invisible matter of the antecedent of $u'$ by $S_{Y}$. For $x$ in $S_{Y}^{-1}(Im(\projection{R'}_Y))$, we define $C_x = \{u | u' = \projection{R'}_Y(x)\}$ as the set of new vertices that want to come from the invisible matter of $x$. It is possible to order $C_x$ by distance from $\projection{R'}_Y(x)$, once again breaking equalities via the lexicographic order. This order is invariant by translation, and as $(\projection{F'},\projection{R'}_{\bullet})$ is bounded, $C_x$ is finite. By writing $C_x = \{u_0, u_1, ...\}$, we can give the origin of the vertices that make it up by fixing, for $i < |C_x|$: $T_{X}(x. (lm)^i) = u_i$ and providing them with an invisible matter tree: $T_{Y}(x.(lm)^irm) = u_imm$ and giving $x$ an intact tree with : $T_{Y}(x.(lm)^{|C_x|}) = \projection{S}_Y(x)mm$. 

We have therefore defined $T_{Y}(x.(lm)^{|C_x|}) = \projection{S}_Y(x)mm$. To complete the definition of $T_{Y}$, we pose $T_{Y}(u.v) = T_{Y}(u).v$ if $u$ is the longest prefix of $u.v$ such that $T_{Y}(u)$ was defined previously.\hfill$\square$
\end{proof}

Unfortunately, it turns out that this simulation does not preserve reversibility, as it cannot simulate a finite set of symmetric graphs. There are therefore a finite number of orbits in $(F,R_{\bullet})$, which are not preserved by the simulation. At first sight, one might hope to restore reversibility by changing the image of a finite number of small graphs with finite invisible matter. Such simple repairs are unfortunately impossible due to the following example:

\newcommand{\symetry}{\text{sym}}

\begin{example}[Contraction-expansion]
Let ${\anonymous{\mathcal X}_{\Sigma,\pi}}$ be a space of graphs and $F$ be a reversible dynamics on $\cal X$ such that there exists $X$ finite with : $$\lim_{n\to\infty} |V(F^n(X))| = \lim_{n\to\infty} |V(F^{-n}(X))| = +\infty$$

The existence of such a dynamics is non-trivial and a central point of \cite{TimeArrow}.

Let $s\notin \pi$, and $\symetry : \anonymous{\mathcal X}_{\Sigma,\pi} \to \anonymous{\mathcal X}_{\Sigma,\pi \cup \{s\}}$ be the operation consisting in duplicating each vertex and connecting each pair of duplicate vertices thus obtained by an edge $\{u:s, u':s\}$. \changes{Let $X$ be an arbitrary graph in $\anonymous{\mathcal X}$.}
 
Let $F'$ be the dynamics on $\anonymous{\mathcal X}\cup \symetry(\anonymous{\mathcal X})$ such that for all $Z$ :
\begin{itemize}
\item $F'(Z)=F(Z)$ and $\symetry(F'(Z))=F(\symetry(Z))$ for all $Z \in \anonymous{\mathcal X}\setminus \{X\}$.
\item $F'(X) = \symetry(F(X))$ and $F'(\symetry(X)) = F(X)$.
\end{itemize}

Such a dynamic is reversible, but it admits at least one disjoint orbits on which it does not preserve the symmetry/asymmetry of the graph. A projection simulation à la Th.~\ref{th:cgd_to_imcgd} of an IMCGD would fail to simulate the dynamics on graphs $X$ and $\symetry(X)$.
\end{example}

We conjecture, however, that it is possible to repair such a simulation by encoding the symmetry of the graphs in the states of the vertices of the associated IMCGD. Thus for a finite number of small anonymous graphs, each symmetric merge/divide of vertices will not be realized in the invisible matter graph, but will be encoded in the state of each vertex.  We call such simulation a “symmetrized projection”. However, formally defining such a simulation requires more work.

\begin{conjecture}
	For any CGD $(F,R_{\bullet})$, there exists an IMCGD $(F',S_{\bullet})$ whose \emph{symmetrized} projection $(F'^{\lrcorner},S^{\lrcorner}_{\bullet})$ is equal to $(F,R_{\bullet})$ for any graph in ${\mathcal X}\cong \projection{\Y}$. Moreover, $(F',S_{\bullet})$ is reversible if and only if $(F,R_{\bullet})$ is reversible.
\end{conjecture}

Such a simulation would allow us to obtain the reversibility of anonymous graph dynamics:

\begin{conjecture}\label{conj:arev}
If an ACGD is invertible, then its inverse function is an ACGD.
\end{conjecture}

\section{Graph dynamics with name algebra}\label{section:ncgd}
\subsection{Definition}

So far we have been working with (pointed) graphs modulo isomorphism, however named graphs turn out to be easier to handle during an implementation and even necessary in the quantum regime \cite{ArrighiQNT}.
In this context, being able to create a new vertex without breaking the locality implies being able to generate a new name from those available locally. For example, if a dynamic $F$ splits a vertex $u$ into two new vertices, the natural choice is to call these vertices $u.l$ and $u.r$. Now, suppose we apply a renaming $R$ transforming $u.l$ into $v$ and $u.r$ into $w$, then applying $F^{-1}$. This time the vertices $v$ and $w$ have merged into one; and in order to maintain invertibility, this one should be named $(v\vee w)$ where $\vee$ is a name-merging operation. However, by choosing trivial $R$, the resulting vertex would have been $(u.l\vee u.r)$, and $F^{-1}\circ F=Id$ requires this vertex to be named $u$. This leads us to a name algebra such as $u=(u.l\vee u.r)$.

\begin{definition}[Name algebra]
Let ${\cal N}$ be a countable set. Consider the terms produced by the grammar ${\cal V} := {\cal N}\,|\,V.\{l , r\}^*\,|\,{\cal V}\vee {\cal V}$ with the equivalence induced by the following rewritings:
\begin{itemize}
\item $u.\varepsilon \longrightarrow u \quad (N)$
\item $(u \vee v).l \longrightarrow u \quad$ and $\quad (u \vee v).r \longrightarrow v\quad(S)$
\item $(u.l \vee u.r)\longrightarrow u\quad(M)$
\end{itemize}

That is, two names $u$ and $v$ are equal if and only if their normal forms $u\!\downarrow_{S\cup M}$ and $v\!\downarrow_{S\cup M}$ are equal.
\end{definition}

This is the algebra corresponding to infinite binary trees. In fact, each element $x$ of ${\cal N}$ can be thought of as a variable representing an infinite binary tree.  The projection operations $.l$ (resp. $.r$) are used to recover the left (resp. right) subtree. The join operation $\vee$ creates a root and then associates the left and right subtrees with it --- this operation is therefore neither commutative nor associative. This infinite binary tree structure is reminiscent of IMCGD, and we will see later that this algebra of names can be seen as an abstraction of invisible matter.
\changesII{
Due to this analogy, we chose to overload the dot notation when working with named graphs. Be careful, however, to distinguish
\begin{itemize}
    \item $v.t$ where $t\in{r,l}^*$, which refers to a fragment of the name $v$.
    \item $v.\alpha$ where $\alpha\in\Pi^*$, which refers to the node reached from $v$, taking path $\alpha$.
\end{itemize}
}
In order to distinguish vertices uniquely, a graph must not have two different vertices with the same name. We must also prohibit the existence of a vertex $x$ and two others named $x.r$ and $x.l$, as these may later merge and cause a name collision.

\begin{definition}[Intersecting]
Two terms $v,v'$ in $\cal V$ are said to be {\em intersecting} if and only if there exists $t,t'$ in $\{l , r\}^*$ such that $v.t=v'.t'$. We use $v\intersectn v'$ as a notation for $\{v.t\,|\,t\in \bint \}\cap \{v'.t'\,|\,t'\in \bint\}$.
%\todo{vérifier quelle est l'utilisation de $v \intersectn V$ }
We also write $v \intersectn V$ for $\bigcup\limits_{v' \in V} v \intersectn v'$.
\end{definition}

\begin{definition}[Well-named graphs]
We say that a graph $G$ is {\em well named} if and only if for any $v,v'$ in $V(G) \subseteq \cal{V}$, $v \intersectn v'\neq \emptyset$ implies $v=v'$.  We denote by ${\cal W}_{\Sigma,\pi}$ or ${\cal W}$ the subset of well-named graphs.
\end{definition}

We now formulate locality conditions for the dynamics over named graphs.
\begin{definition}[Uniform continuity]
A function $F$ on ${\cal W}$ is said to be \changes{{\em uniformly continuous}} if and only if for all $m\geq 0$, there exists $n\geq 0$, such that for all \changes{$G \in \mathcal{W}$, for all $v \in V(G)$ and $v' \in V(F(G))$, $v \intersectn v' \neq \emptyset$} implies $(F(G))^m_{v'}=(F(G^n_{v}))^m_{v'}$.
\end{definition}
\changes{Notice that here we have expressed locality as uniform continuity, meaning that the radius $n$ is not allowed to depend on $G$. For CGD we had used continuity instead, seemingly allowing for $n$ to depend on $X$ (Def. \ref{def:continuitymodulo}). But in that case the space was compact, so uniform continuity and continuity were in fact equivalent (Heine–Cantor theorem). Continuity was preferred mainly in order to bring us one step closer to reversibility results.\\}

\acom{
By contradiction, let us assume that a translation invariant, continuous, bounded and name preserving $F$ over \mathcal{W}, is not uniformly continuous.
Therefore, $\exists n \geq 0$ such that $\forall m\geq 0, \exists G_m\in \mathcal{W}, \exists v_m\in V(G)$ and $\exists v_m' \in V(F(G))$ such that $v_m \ \intersectn v'_m \neq \emptyset$ and $(F(G_m))^n_{v_m'}\neq(F(G^m_{m.v_m}))^n_{v_m'}$. 

The sequence of graph modulo isomorphism $(\anonymous{G_m})_{m\in \N}$ over $\mathcal{X}$ admit a converging subsequence $(\anonymous{G_{m_i}})_{i\in \N}$ as $\mathcal{X}$ is compact. Let $X$ be the limit of that subsequence and let $G \in \mathcal{W}$ such that $\anonymous{(G,v)}=X$. %blabla sous sequence
We have that for all $i \in \N$, there exist a renaming $R_i$ such that $(G_{m_i.v_i})^{m_i}= R_i(G_{R_i^{-1}(v)}^{m_i})$. this contradict our hypothesis as for all $i \geq m$ we have :
\begin{align*}
     F(G^m_v) 
     &= F((G^{m_i}_v)^m_v)\\
     &= F((R_iR_i^{-1}G_v^{m_i})^m_v)\\ 
     &= F((R_iG_{R_i^{-1}v}^{m_i})^m_v)\\
     &= F(((G_{m_i.v_i})^{m_i})^m_v) \\
     &\neq (F(G^m_{m_i.v_m}))^n_{v_{m_i}'}\\
     &= (F(G^m_{m_i.v_m}))^n_{v_{m_i}'}
\end{align*}

\begin{align*}
     F(G)_{v'}^n
     &=  F(G)_{v'}^n
     
\end{align*}

Let $v' \in V(F(G))$ such that $v \intersectn v' \neq \emptyset$, by continuity of $F$ we have that there exist $n$ $F(G^m_v) = F(G^m_v)_{v'}^n$ but
}

We now formulate the equivalent of the expansivity bound for dynamics over named graphs. \changesII{The idea is that nodes should not be broken any further than some depth $b$. We formalize this by stating that any fragment of depth further than $b$ of a node of the original graph has to be found ``as is'' in the image graph.}

\begin{definition}[Bounded dynamics]\label{def:boundednessfornamed}
A function $F$ on ${\cal W}$ is said to be {\em bounded} if and only if there exists a bound $b$ such that for all $G$, for all $v \in V(G)$, for all $s \in \{l,r\}^*$ such that $|s|\geq b$, $\exists v' \in V(F(G))$ and $s' \in \{l,r\}^*$ such that $v.s = v'.s'$.
\end{definition}

Once again, we want our dynamic to be independent of the naming, so we ask that it commute with renamings that do not create an intersection between the names :

\begin{definition}[Renamings]\label{def:graph renaming}
Let $R$ be an injective function from ${\cal N}$ to ${\cal V}$ such that for any $x\neq y\in {\cal N}$, $R(x)$ and $R(y)$ do not intersect. The natural extension of $R$ to the set ${\cal V}$, according to $$R(u.l)=R(u).l\qquad R(u.r)=R(u).r\qquad R(u\vee  v)=R(u)\vee R(v)$$ is called a "renaming".
\end{definition}

\begin{definition}[Invariance by translation]
		A function $F$ on ${\cal W}$ is said to be {\em invariant by translation} if and only if for all $G\in {\cal W}$ and for all renaming $R$, $F(RG)=RF(G)$.
\end{definition}

Finally, we want our dynamics to preserve the name space so that we can track the evolution of vertices.
\begin{definition}[Preservation of names]
Let $F$ be a function on ${\cal W}$. We say that the function $F$ {\em preserves names} if and only if for all $u$ in ${\cal V}$ and $G$ in ${\cal W}$ we have that $u\intersectn V(G) = u\intersectn V(F(G))$.
\end{definition}

Finally, we can state what named causal graph dynamics are.

\begin{definition}[Named Causal Graph Dynamics (NCGD)]\label{def:NCGD}
A function $F$ on ${\cal W}$ is said to be a {\em Named Causal Graph Dynamics (NCGD)} if and only if it is translation invariant, \changes{uniformly} continuous, bounded and name-preserving.
\end{definition}

\changes{
\begin{conjecture}[Continuity same as uniform continuity for NCGD]
In Def. \ref{def:NCGD}, we may equivalently replace uniform continuity by continuity, i.e. the property that for all $G \in \mathcal{W}$ and $n\geq 0$, there exists $m\geq 0$, such that for all $v \in V(G)$ and $v' \in V(F(G))$, $v \intersectn v' \neq \emptyset$ implies $(F(G))^n_{v'}=(F(G^m_{v}))^n_{v'}$.  
\end{conjecture}
Indeed, albeit names jeopardize the compactness of the graph space, forbidding us to apply Heine-Cantor out-of-the-shelf, renaming-invariance make these names innocuous to some extent. A proof of a slightly more restrictive version of this conjecture has been given in \cite{ArrighiCGD} as Theorem 3.
}
\begin{conjecture}\label{conj:nrev}
If an NCGD is invertible, then its inverse function is an NCGD.
 \end{conjecture}
\changes{A proof of a slightly more restrictive version of this conjecture has been given in \cite{ArrighiCGD} as Theorem 4. 
}

Once again, we make sure that the creation and destruction of vertices is compatible with invertibility (we'll see later that it corresponds to reversibility):

\begin{figure}
\begin{center}
\includegraphics[scale=1,clip=true,trim=0cm 0.2cm 0cm 0.2cm]{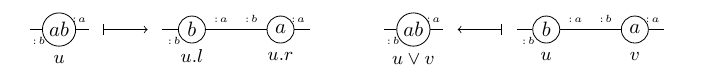}
\end{center}
\caption{\label{fig:GnAScattering} {\em The scattering step of Hasslacher-Meyer as implemented by an NCGD.}}
\end{figure}

\begin{example}[Named Hasslacher-Meyer]\label{ex:NHM}
Let {\cal W} be with ports and states as in the anonymous Example \ref{ex:genadv}. Once again, we alternate the advection step with the collision step explained in figure \ref{fig:GnAScattering}.
\end{example}

\changes{
\begin{theorem}[Composability of NCGD]
    For all $F,F'$ NCGD over $\cal W$, the function $F' \circ F$ is a also an NCGD over $\cal W$.
\end{theorem}

\begin{proof}
Let $F$ and $F'$ two NCGD over $\cal W$. First, notice that $F' \circ F$ is well defined as $F(\cal W)\subseteq\cal W$.
We can also see that $F' \circ F$ is name preserving, as for $u \in \cal V$ and $G \in \cal W$, $u \intersectn V(G) = u \intersectn V(F(G)) = u \intersectn V(F'(F(G)))$.

Similarly we have that $F'\circ F$ is translation-invariant, as for all renaming $R$, we have that $F' \circ F(RG) = F'(RF(G)) = R(F' \circ F(G))$.

Now let us prove that $F'\circ F$ is bounded for the bound $b + b'$ where $b$ and $b'$ are the $F$ and $F'$ bounds respectively. Let $v \in V(G)$ and $r\in  \{l,r\}^*$ with $|r|\geq b$. Let $s,t\in \{l,r\}^*$ such that $r=st$, $|s|\geq b$ and $|t|\geq b'$. As $F$ is bounded, there is $v'\in V(F((G))$ and $s' \in \{l,r\}^*$  such that $v.s = v'.s'$.
Let $r'=s't$ and notice that $|r'|\geq|t|\geq b'$, thus there exists $v'' \in V(F'\circ F(G))$ and $r''\in \{l,r\}^*$ such that :
$$v''.r''=v'.r'= v'.s't=v.st= v.r$$

Finally, let us prove that $F'\circ F$ is uniformly continuous. Let $n \geq 0$, and $m'$ the uniform continuity of $F'$ for such $n$. Now let $m$ the uniform continuity radius of $F$ for $m'$. Let $v \in V(G)$ and $v''\in F'(F(G))$ such that $v \intersectn v'' \neq \emptyset$. As $F$ and $F'$ are name preserving, there exists $v' \in V(F(G))$ such that $v' \intersectn (v \intersectn v'') \neq \emptyset$ and therefore $v \intersectn v'\neq \emptyset$ and $v' \intersectn v''\neq \emptyset$. Applying both uniform continuities, we obtain the following equalities:

\begin{align}
    F'(F(G))^{n}_{v''} &= F'(F(G)_{v'}^{m'})^{n}_{v''} \\
    &= F'(F(G_v^m)_{v'}^{m'})^{n}_{v''} \\
    &= F'(F(G_v^m))^{n}_{v''}
\end{align}
where the equalities $(1)$ and $(3)$ are obtained from the uniform continuity of $F'$ and $(2)$ is obtained from the uniform continuity of $F$.
\end{proof}
}

\subsection{Simulation of IMCGDs by NCGDs}
\newcommand{\nameinduced}[1]{{{#1}^\#}}
\newcommand{\renaming}{S}
\changes{In order to define an NCGD $\nameinduced{F}(G)$ that simulates an IMCGD $F(X)$, the idea is as follows. First we name the nodes of $X$, starting from its the pointer, in order to obtain $G$. Then, by tracing the names of $G$ through $R_\bullet$, we name the nodes of $F(X)$ and describe $\nameinduced{F}(G)$.}

\begin{definition}[Naming function, compatibility]\label{def:namingfunction}
Consider a function $N:\Pi^*\rightarrow \mathcal{V}$. It is called a {\em naming function} if and only if for all $u,v \in \Pi^*$, $N(u)\intersectn N(v)\neq \emptyset$ implies $N(u)=N(v)$. For any $X \in \mathcal{X}$, we say that $N$ is {\em compatible with} $X$ if we have that for all $u,v \in V(X)$, $N(u) = N(v)$ if and only if $u$ and $v$ designate the same vertex in $X$. We denote by $G=N(X)$ the well-named graph obtained by applying $N$ to every vertex of $X$.
\end{definition}

\begin{lemma}{\label{lem:renaming_naming}}
    Consider a naming function $N:\Pi^*\rightarrow \mathcal{V}$ that is compatible with $X \in \mathcal{X}$, and a renaming $S : \mathcal{V}\rightarrow \mathcal{V}$. Then $M = S \circ N$ is a naming function that is compatible with $X$. 
\end{lemma}
\begin{proof}
$[$Naming function$]$ Consider $u,v \in \Pi^*$ such that $S(N(u))\intersectn S(N(v))\neq \emptyset$. There exists $t,t' \in \{l,r\}^*$ such that $S(N(u)).t = S(N(v)).t'$ and thus $S(N(u).t) = S(N(v).t')$, from which we have $N(u).t=N(v).t'$, i.e. $N(u)\intersectn N(v)\neq \emptyset$. Since $N$ is a compatible naming function, $u=v$, and so $S\circ N$ is a naming function.\\
$[$Compatibility$]$ Let $u,v \in V(X)$. If $u$ and $v$ designate the same vertex, then $N(u) = N(v)$, and thus $S(N(u)) = S(N(v))$.
If $u$ and $v$ do not designate the same vertex, then we have $N(u) \neq N(v)$. Since $S$ is injective, this implies that $S(N(u)) \neq S(N(v))$.
\end{proof}

\begin{definition}[De-anonymization]{\label{def:deanonymization}}
  Let $(F,R_\bullet)$ be an IMCGD over $\mathcal{Y}$. Let $\mathcal{W}'\subseteq \mathcal{W}$ be the set of well-named graphs such that $G \in \mathcal{W}'$ if and only if there exist $Y \in \mathcal{Y}$ and a compatible naming function $N$ such that $G=N(Y)$. 
  
  The dynamics induced on the named graphs $\nameinduced{F}$ is defined such that for all $G \in \mathcal{W}'$, for all $Y \in \mathcal{Y}$ and naming function $N$ such that $G=N(Y)$, we have
  $$\nameinduced{F}(G) = M(F(Y))$$ 
  with $M$ the naming function $N\circ R_Y^{-1}$. 
\end{definition}
\begin{remark}
    Let us prove that $\nameinduced{F}$ is well defined, i.e. that it does not depend on the choices of $Y$ and $N$ that produce $G=N(Y)$. Let $G \in \mathcal{W}, Y,Y' \in \mathcal{Y}$ and two compatible naming function $N,N'$ such that $G=N(Y)=N'(Y')$. By Def.~\ref{def:namingfunction} the graphs $Y, Y'$ must be isomorphic, else the two compatible namings would not be equal everywhere.
    Since $Y$ and $Y'$ are isomorphic to each other, there exists $u\in V(Y)$ such that $Y'=Y_u$. Combining this with the fact that the same $G$ is produced, we have that for any $v \in Y$, $N'(\reversepath{u}v) = N(v)$.\\
    Let $M=N\circ R_Y^{-1}$ and $M'=N'\circ R_{Y'}^{-1}$, we need to prove that $M(F(Y))=M'(F(Y'))$. It suffices to prove that for any $w \in F(Y)$ and letting $w'=\reversepath{R_Y(u)}.w\in F(Y')$, we have $M(w)=M'(w')$. Equivalently, $N(R_Y^{-1}(w)) = N'(R^{-1}_{Y_u}(\reversepath{R_Y(u)}.w))$.\\
    For all $w\in V(F(Y))$, consider $v=R_Y^{-1}(w)\in V(Y)$ and express the shift invariance of $R_\bullet$: 
    \begin{align*}
        R_Y(u).R_{Y_u}(\reversepath{u}v)&=R_Y(v)\\
        R_{Y_u}(\reversepath{u}v) &= \reversepath{R_Y(u)}R_Y(v)\\
        \reversepath{u}v &= R_{Y_u}^{-1}(\reversepath{R_Y(u)}R_Y(v))\\
        N'(\reversepath{u}v) &= N'(R_{Y_u}^{-1}(\reversepath{R_Y(u)}R_Y(v)))\\
        N(v) &= N'(R_{Y_u}^{-1}(\reversepath{R_Y(u)}R_Y(v)))\\
    \end{align*}
    Substituting back $v=R_Y^{-1}(w)$, concludes our proof.
\end{remark}

    %

    %This prove that both $H$ and $H'$ are a naming of $F(Y)$ and therefore are isomorph. Now let us consider a vertex $v \in V(F(Y'))$, by shift invariance we have $R_X(u.v)=R_X(u).R_{X_u}(v)$
    %\begin{align*}
    %    N' \circ R_{Y'}^{-1}(v) &= N' \circ R_{Y_u}^{-1}(v)\\
    %\end{align*}

% \begin{definition}
%   Let $(F,R_\bullet)$ be an IMCGD over $\mathcal{Y}$. Let $\mathcal{W}'\subseteq \mathcal{W}$ be the set of well-named graphs such that $G \in \mathcal{W}'$ if and only if $\anonymous{G} \in \mathcal{Y}$. The dynamics induced on the named graphs $\nameinduced{F}$ is defined such that for all $G \in \mathcal{W}'$ :
  
%   \begin{itemize}
%   \item For any vertex $u\in V(G)$, $\anonymous{(\nameinduced{F}(G),u)}=F(\anonymous{(G,u)})$
%   \item For all vertices $u,v\in V(G)$, if $p(v)$ is a path from $u$ to $v$, then the vertex obtained by following $R_{\anonymous{(G,u)}}(p(v))$ in $F(\anonymous{(G,u)})$ is named $v$.
%   \end{itemize}
% \end{definition}

%\todo{Clarifier def 27}
Intuitively, an IMCGD $(F,R_\bullet)$ is therefore simulated by naming the whole graph in such a way as to obtain a well-named graph, and keeping the naming consistent with $R$. Since $R$ is a bijection, it is not even necessary to create/destroy new nodes.

\changes{
This notion of dynamics induced on the named graphs composes well:
\begin{lemma}[Composability of induced NCGD]{\label{lem:inducedcomposability}}
  Let $(F,R_\bullet)$ and $(F',S_\bullet)$ be two IMCGD over $\mathcal{Y}$. Let $\nameinduced{F}$, $\nameinduced{F'}$ be the corresponding dynamics induced on the named graphs. Then $\nameinduced{F'}\circ \nameinduced{F}$ is also the dynamics induced on the named graphs of the IMCGD $(F'\circ F,T_\bullet)$ with $T_Y=S_{F(Y)}\circ R_Y$. 
\end{lemma}
\begin{proof}
In Lemma \ref{lem:im-composability} we proved that $(F'\circ F,T_\bullet)$, is indeed an IMCGD. 
We need to prove that for all $G \in \mathcal{W}'$ for all $Y \in \mathcal{Y}$ and naming function $N$ compatible with $Y$ such that $G=N(Y)$, we have $(\nameinduced{F'}\circ \nameinduced{F})(G) = M(F(Y))$ with $M$ the naming function $N\circ T_Y^{-1}$. Indeed,
\begin{align*}
\nameinduced{F'}(\nameinduced{F}(G)) &=\nameinduced{F'}((N\circ R_Y^{-1})(F(Y)))\\
&=(N\circ R_Y^{-1}\circ S_{F(Y)}^{-1})(F'(F(Y)))\\
&=(N\circ (S_{F(Y)} \circ R_Y)^{-1})(F'(F(Y)))  
\end{align*}
\end{proof}
}

\begin{theorem}\label{th:imcgd_to_ncgd}
The dynamics induced by an IMCGD on named graphs is an NCGD.
\end{theorem}

\begin{proof}
We need to prove that the dynamics induced on the named graphs is name-preserving, bounded, shift-invariant, and uniformly continuous.\\
Name-preservation is easily deduced from the fact that an IMCGD preserves names in an even stricter way: no name is created nor destroyed. For the same reason, boundedness is trivial as names are left unchanged.\\
Now let us focus on shift-invariance. Let $G \in \mathcal{G}$. Let $Y \in \mathcal{Y}$ and a naming function $N$ compatible with $Y$ such that $N(Y)=G$. Let $S$ be a renaming. As $N' = S \circ N$ is also a naming function compatible with $Y$ by Lem. \ref{lem:renaming_naming}, we have:
$$\nameinduced{F}(SG) = S \circ N \circ R_Y^{-1}(Y) = S(\nameinduced{F}(G)).$$
Finally, let us consider uniform continuity. 
\changes{First of all notice that if $v\in V(G)$ and $v'\in V(\nameinduced{F}(G))$ are such that $v \intersectn v'\neq \emptyset$, then, because $G$ is well-named and $\nameinduced{F}$ is name-preserving, $v=v'$. Thus, uniform continuity reduces to for all $m$, there exists $n$ such that for all $v$, $\nameinduced{F}(G)^m_v=\nameinduced{F}(G^n_v)^m_v$: this is what we have to prove.} Consider a naming function $N$ that is compatible with $Y\in \mathcal{Y}$ and such that $N(\varepsilon)=v$ and let $G:=N(Y)$. Without loss of generality by Def.~\ref{def:deanonymization}, we have that $\nameinduced{F}(G)=M(F(Y))$ with $M=N\circ R_Y^{-1}$. But notice that naming functions essentially commute with the disk operation, so that that we have $\nameinduced{F}(G)^m_v=M(F(Y)^m)$.\\
Similarly, we have $\nameinduced{F}(G^n_v)=M'(F(Y^n))$ with $M'=N\circ R_{Y^n}^{-1}$, and 
$\nameinduced{F}(G^n_v)^n_v=M'(F(Y^n)^m)$.\\
Now, by corollary \ref{cor:IMQ_compactness}, $F$ is uniformly continuous, and thus we have that $F(Y^n)^m=F(Y)^m$. So, one the one hand 
$\nameinduced{F}(G)^m_v=M(F(Y)^m)$
and on the other hand 
$\nameinduced{F}(G^n_v)^m_v=M'(F(Y)^m)$.
Thus, we will have shown our result if $M$ and $M'$ coincide over any $u\in V(F(Y)^m)=V(F(Y^n)^m)$. To prove this we start from the uniform continuity of $R_\bullet$. It says that when restricted to the co-domain $V(F(Y)^m)=V(F(Y^n)^m)$ we have that 
\begin{align*}
R_Y&=R_{Y^n}\\ 
R_Y^{-1}&=R_{Y^n}^{-1}\\
N\circ R_Y^{-1}&=N\circ R_{Y^n}^{-1}\\
M&=M'
\end{align*}
We thus have $\nameinduced{F}(G)^m_v=\nameinduced{F}(G^n_v)^m_v$.

%Recall that by continuity of $F$, we have for all $Y$, for all $n$, there exists $m$ such that $(F(Y))^m=(F(Y^n))^m \textrm{, }\dom\,R_{Y}^m\subseteq V(Y^n)\textrm{ and }R_{Y}^m=R_{Y^n}^m$. Let us consider $v$ the name such that $N(\varepsilon) = v$. We have the following equalities :
%\begin{align*}
%&(\nameinduced{F}(G))^m_v = (N(F(Y)))^m_v = N(F(Y)^m) =  N(F(Y^n)^m)\\ 
%&= (N(F(Y^n)))^m_v = (\nameinduced{F}(N(Y^n))^m_v  =  (\nameinduced{F}(G^n_v))^m_v 
%\end{align*}
%Notice that the second to last equality only works because the continuity of $(F,R_\bullet)$ entails $R_{Y^n}^m=R_{Y}^m$ and therefore $(N(F(Y^n)))^m_v$ and $G^m_v=N(Y)^m_v$ share the same set of names.
\end{proof}

%\begin{proof}
%We need to prove that the induced named dynamics is name-preserving,  shift-invariant and continuous. 

%Let us consider a renaming $\renaming$. First, we can see that $\anonymous{(G,u)}=\anonymous{(\renaming G,\renaming (u))}$.
% For all $v \in V(G)$, the image of $v$ in $F(G)$ is the vertex obtained by following the path $R(v)$ from $u$. Since paths are invariant under renaming, this is the same vertex obtained by following $R(v)$ from $\renaming(u)$, so we have that $\renaming \nameinduced{F}(G)= \nameinduced{F}(\renaming G)$. This concludes shift-invariance.
 
%By continuity of $(F,R_\bullet)$ we have that for all $m$, there exists $n$ such that: $$\anonymous{{(\nameinduced{F}(G),u)^m}} = (F(Y))^m = (F(Y^n))^m= \anonymous{{(\nameinduced{F}(G^n_u),u)^m}}.$$
%  Now if we focus on names, by continuity we also have $\text{dom } R_Y^m \in V(Y^n)$ so for all $v \in {(\nameinduced{F}(G),u)}^m$, we have $v \in V(G^n_u)$. This concludes the continuity of $\nameinduced{F}$ because $(\nameinduced{F}(G))^m_u=(\nameinduced{F}(G^n_u))^m_u$. \hfill$\square$\end{proof}

\begin{theorem}
Consider $(F,R_\bullet)$ an IMCGD and $\nameinduced{F}$ its dynamics induced on the named graphs. $(F,R_\bullet)$ is invertible if and only if $\nameinduced{F}$ is invertible.
\end{theorem}

\begin{proof}
$[\implies]$ First, suppose $(F,R_\bullet)$ is invertible, and let $({F}^{-1},S_\bullet)$ be its inverse dynamic.

Let us consider $\nameinduced{(F^{-1})}$ as a potential inverse for $\nameinduced{F}$. Let $G\in \mathcal{W'}, Y \in \mathcal{Y}$ and a naming function $N$ such that $G=N(Y)$. Recall that $S_{F(Y)} \circ R_Y=Id$. We have: 

$$\nameinduced{F} \circ \nameinduced{(F^{-1})}(G) = ((N \circ R^{-1}_Y) \circ S^{-1}_{F(Y)})(Y) = N(Y) = G$$

Similarly we have $R_Y\circ S_{F(Y)}=Id$ and thus
$$\nameinduced{(F^{-1})}\circ \nameinduced{F} (G) = ((N \circ S^{-1}_{F(Y)}) \circ R^{-1}_Y) (Y) = N(Y) = G$$

Therefore $\nameinduced{F}$ is invertible and its inverse is $\nameinduced{(F^{-1})}$.

$[\impliedby]$ Now suppose that $\nameinduced{F}$ is invertible. For all $Y \in \mathcal{Y}$, consider a compatible naming function $M$, such that $M(\varepsilon) = 0$ where $0$ is some distinguished element of $\mathcal{V}$. Let $F^*$ be the function such that 
$$F^*(Y)= \anonymous{((\nameinduced{F})^{-1}(M(Y)),0)}.$$

%%%%%%%%%%%%%%%%%%%%%%%%%%
Let us prove that for all $Y\in {\cal Y}$, $F^*(F(Y))=Y$. To this end consider the naming function $N:=M\circ R_Y$, so that $M=N\circ R^{-1}_Y$, and observe that $N(\varepsilon)=M(R_Y((\varepsilon)))=0$. Let $G=N(Y)$ and notice immediately that $\anonymous{(G,0)}=Y$. By definition of $\nameinduced{F}$, we have that $\nameinduced{F}(G)=M(F(Y))$. Since $\nameinduced{F}$ is invertible, $G=\nameinduced{F}^{-1}(M(F(Y)))$. Consequently, $Y=\anonymous{(\nameinduced{F}^{-1}(M(F(Y))),0)}$. By definition of $F^*$, this is indeed $F^*(F(Y))=Y$.\\
%%%%%%%%%%%%%%%%%%%%%%%%%%%%%%%%%%%%%%%%%%%%%%%%%%%%%%%%
Let us prove that for all $Y\in {\cal Y}$, $F(F^*(Y))=Y$. Notice that for any graph $Y' \in Y$ and any compatible naming function $M$, such that $M(\varepsilon) = 0$, we have that $F(Y')=\anonymous{\nameinduced{F}(M(Y')),0)}$. Note that for such an $M$, $M(Y) = (\anonymous{(M(Y),0))}$ and applying it to the definition of $F^*$ we obtain that 
$M(F^*(Y))={\nameinduced{F}}^{-1}(M(Y))$.

\begin{align*}
(F \circ F^*)(Y) &= (\anonymous{\nameinduced{F}(M(F^*(Y)),0)}\\ 
&= (\anonymous{(\nameinduced{F}\circ \nameinduced{F}^{-1})(M(Y)),0)}\\
&= \anonymous{(M(Y),0)}\\ 
&= Y    
\end{align*}

%Let $F^{-1}$ be the function such that $Y=\anonymous{(G,u)}$, $F^{-1}(Y) = \anonymous{({\nameinduced{F}}^{-1}(G),u)}$.
 %We have the following equalities:
 %$$ F\circ F^{-1}(Y)= \anonymous{(\nameinduced{F} \circ {\nameinduced{F}}^{-1}(G),u)} =  \anonymous {(G,u)} = Y$$

 %Similarly, we have :
 % $$ F^{-1}\circ F(Y)= \anonymous{({\nameinduced{F}}^{-1} \circ \nameinduced{F}(G),u)} =  \anonymous {(G,u)} = Y$$
 %$F\circ F^{-1} = F^{-1} \circ F  = \text{Id}$, so $F$ is invertible which conclude the right to left implication.

 %\acom{
% Now, assume that $(F,R_\bullet)$ is invertible, and let $F^{-1}$ be the inverse function of $F$.
%Let $\nameinduced{F^{-1}}$ be the NCGD induced by $(F^{-1},R_\bullet^{-1})$. For all $G \in \mathcal{W}$ and $u \in V(G)$, we have :

%\begin{align*}
%     \anonymous{(G,u)}&=Y &= F \circ F^{-1}(Y)  &= F \circ F^{-1}(\anonymous{(G,u)})&= \anonymous{(\nameinduced{F} \circ {\nameinduced{F^{-1}}}(G),u)}\\
%     & &= F^{-1} \circ F(Y) &= F^{-1} \circ F(\anonymous{(G,u)}) &= \anonymous{(\nameinduced{F^{-1}} \circ {\nameinduced{F}}(G),u)}
%\end{align*}

%So there exists a renaming $\renaming$ such that $\nameinduced{F}\circ {\nameinduced{F^{-1}}(G)}=\renaming(G)$. 
% Let ${\nameinduced{F}}^{-1}={\nameinduced{F^{-1}}}\circ\renaming^{-1}$ ....
 %But as proved earlier, $\nameinduced{F}$ is an NCGD and therefore translation invariant, which allows us to conclude $\nameinduced{F} \circ {\nameinduced{F}}^{-1}= \text{ Id}$. \hfill$\square$
% }
\end{proof}

\subsection{Simulation of NCGDs by IMCGDs}

We will now show that NCGDs can be simulated by IMCGDs. The goal is to induce, from an NCGD, a CGD on graphs with invisible matter such that $R_\bullet$ is a bijection. The main idea is therefore to use names to establish $R_\bullet$. The question then arises of how to deal with invisible matter, given that this is not present in named graphs. The first step is therefore to attach a tree of invisible matter to each vertex of $G$, then to name the new vertices according to the root. Note that the resulting graph $G^\vartriangle$ is not well-named, but this is only an intermediate step necessary for the construction of $R_\bullet$. Once this construction has been obtained, all we have to do is "forget" the names by considering graphs modulo isomorphisms.

\begin{definition}[Named invisible matter graphs]
Let $G \in {\cal W}$ be a well named graph. Construct its {\em associated graph} $G^\vartriangle$ by attaching a tree of invisible matter to each vertex, and naming each vertex in the invisible matter according to the convention in figure \ref{im:naming_map.}. More precisely, if $u$ is a visible vertex and $t \in mm.\{lm,rm\}^*$ in the invisible matter, the vertex obtained by traversing $t$ from $u$ is named $u.\eta(t)$, where $\eta : mm\IM \rightarrow \{l,r\}^*$ is the function such that:
 \begin{itemize}
 \item $\eta(mms)=r^{n+1}$ if $s=(rm)^n$,
 \item $\eta(mms)=s'$, where $s'$ is the word obtained by removing ports $m$ from $s$, otherwise.
\end{itemize}
\end{definition}

\begin{figure}[h]
\begin{center}
 \includegraphics[scale=1.2]{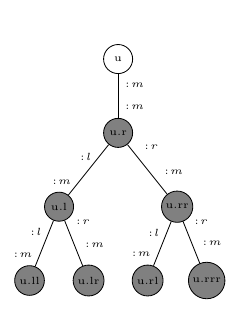}\\
\end{center}
 \caption{\label{im:naming_map.} Invisible matter's naming scheme.}
\end{figure}

The behaviour of an NCGD $F^\vartriangle$ on a named invisible matter graph naturally induces an IMCGD $(F,R_\bullet)$ in the special case of $V(G)=V(F^\vartriangle(G))$. On the other hand, the question becomes unclear when $F^\vartriangle$ performs mergers and separations. In order to track down names through separations and merges, we will use the following functions:

\begin{definition}[Separation of names]
Consider $u \in {\cal V}$ hold. We define $\nu_u : {\cal V} \rightarrow {\cal V}$ the function such that :
\begin{itemize}
 \item $\nu_u(u)= u.l$
 \item $\nu_u(u.lr^n)=u.lr^{n+1}$
 \item $\nu_u(v.t)=v.t$, otherwise.
\end{itemize}
If $A$ is a set of names, we write $\nu_u(A)=\{\nu_{u}(a)\,|\,a\in A\}$. Figure \ref{fig:sigma} shows how this function allows us to track down names when splitting an invisible matter tree in two.
\end{definition}

\begin{figure}[h]
\begin{center}
 \includegraphics[scale=1.2]{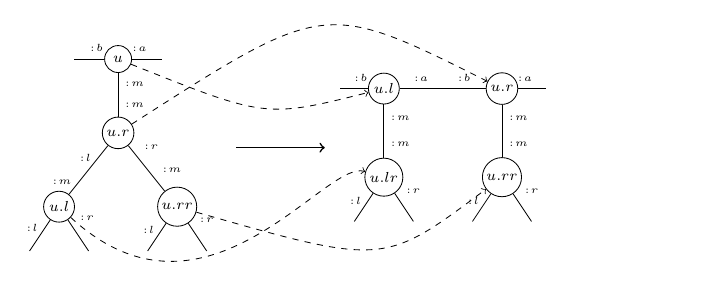}\\
\end{center}
 \caption{\label{fig:sigma} The splitting of a $u.\bint$ tree by $\nu_u$.}
\end{figure}

This operation can be interpreted as changing the name $u$ into $u.l$ and, in order to preserve injectivity, translating the branch $lr^n$. Typically this happens when a root tree $u$ is split into two root trees $u.l$ and $u.r$, as in Fig. \ref{fig:sigma}. The following lemma shows that the names of $G$ and $F(G)$ can always be aligned using such functions $\nu_\bdot$.

\newcommand{\iminduced}[1]{#1^{\vartriangle}}

\begin{lemma}\label{lem:sigmas}
Let $U \subseteq V(G)$ and $U' \subseteq V(\iminduced F(G))$. There exists \changes{$b$, $m\leq b$, $n\leq b$,} $S= \nu_{u_1} \circ ... \circ \nu_{u_n}$ and $S'= \nu_{u'_1} \circ ... \circ \nu_{u'_m}$, such that for all $v \in S(U), v' \in S'(U')$, $v \intersectn v' = \emptyset$ or $v=v'$.
\end{lemma}
\begin{proof}
We decompose the names of $U$ and $U'$ using $\nu$ until they become equal, using the following algorithm:
\begin{algorithmic}
 \State{$S:=\mathbf{I}$}
 \State{$S':=\mathbf{I}$}
  \While{$\exists u \in S(U)$, $\exists {u'}\in S'(U')$, $t,t'=\bint$ 
  such that $u.t=u'.t'$ is the longest common suffix of $t$ and $t'$ is $\varepsilon$}.
  \If{$|t|>|t'|$}
    \State{$S:=\nu_{u}\circ S$}
  \Else{}
    \State{$S':=\nu_{{u'}}\circ S'$}
   \EndIf
  \EndWhile
\end{algorithmic}
\changes{
Let $b'$ be the bound for $\iminduced F(G)$ as given by Def. \ref{def:boundednessfornamed}.
For $n=0$, let $m$ be its radius of uniform continuity. Let $b={|\pi|}^{m+1}2^{b'}.$ Let us prove that $m\leq b$ and $n\leq b$. 
By the choice of $b'$, we know that for all $u\in V(G)$, for all $s\in\{l,r\}^*$ such that \changesII{$|s|\geq b'$}, there exists $u'\in \iminduced F(G)$, $s'\in \{l,r\}^*$ such that $u'.s'=u.s$. Then $b'$ is a bound for $n$. Since $\iminduced F(G)$ is name-preserving, it follows that for all $u'\in V(\iminduced F(G))$, $u'$ is generated by the grammar
$$w,w'::= u_i.s_i \mid w \vee w'$$
with $u_i\in V(G)$ and $s_i\in\{l,r\}^*$ with $|s_i|\leq b'$. Let $u_0.s_0$ be one such fragment of $u'$, by uniform continuity, for all $i$, $u_i\in V(G_{u_0}^r)$. The number of such $u_i$ can be bounded by $|\pi|^{r+1}$. The number of $u_i.s_i$ fragments can be bounded by $|\pi|^{r+1}2^{b'}$. The dept of $u'$, before a $u_i.s_i$ leaf is encountered, can therefore be bounded by $b$. This bounds $m$ by $b'$.}
\hfill$\square$
\end{proof}

So we can now deduce a 'successor' notion for a vertex of a named invisible matter graph evolving through an NCGD.

\begin{definition}[Induced naming]
Let $\iminduced{F}$ be an NCGD. For all $G$, we define its {\em induced naming } $\iminduced{R}_{G}$ as follows. For all $u.t \in V(G).\bint$, for all $u'.t' \in V(\iminduced{F}(G)).\bint$, we have $\iminduced{R}_{G}(u.t)=u'.t'$ if and only if $S(u.t)=S'(u'.t')$, where $S$ and $S'$ are the results of applying the previous Lemma \ref{lem:sigmas} to $\{u\}$ and $\{u'\}$.
\end{definition}

\begin{remark}
This definition is well-founded, i.e. for all $u \in V(G)$ and $t \in \bint$, there is at most one $v$ and $s$ such that $S(u.t)=S'(v.s)$ where $S$ and $S'$ are the results of applying Lemma \ref{lem:sigmas}, see Lemma \ref{lem:soundness_name_map} for the formal proof.
\end{remark}

Now that the NCGD has been extended to act on trees of invisible matter, and it is possible to track the evolution of each of the vertices, all we have to do is forget the names to obtain an IMCGD.

\begin{definition}[Induced dynamics]\label{def:induced_dyn}
Let $\iminduced{F}$ be an NCGD. Its induced dynamics on the invisible matter graphs $(F,R_\bullet)$ is such that for all $G$ and $u.t \in V(G).\bint$ :
 \begin{itemize}
  \item $F(\anonymous{(G^\vartriangle,u.t)})=\anonymous{(\iminduced{F}(G)^\vartriangle,\iminduced{R}_G(u.t) )}$.
  \item $R_{\anonymous{(G^\vartriangle,u.t)}}(p)$ is the path between $\iminduced{R}_G(u.t)$ and $\iminduced{R}_G(v.s)$ in $\iminduced{F}(G)^\vartriangle$, where $v.s$ is obtained by following the path $p$ from $u.t$ in $G^\vartriangle$.  
 \end{itemize}
\end{definition}

\newcommand{\Feta}[2]{\anonymous{(#1^\vartriangle,#2)}}
\begin{theorem}\label{th:ncgd_to_imcgd}
 The dynamics induced by an NCGD is an IMCGD.
\end{theorem}

\begin{proof}
We need to prove that the induced dynamics is translation invariant, continuous, quiescent on invisible matter and vertex-preserving.

Let us first concentrate on translation invariance. Let $Y\in \Y$ be an invisible matter graph, $G\in \mathcal W$ and $a \in V(G).\bint$ be such that $Y=\Feta{G}{a}$. Let $u$ be the path from $a$, and $b$ the vertex obtained by following $u$ from $a$. By definition of $(F,R_\bullet)$ we have the following equalities:
 $$F(Y_u)= F(\Feta{G}{b})=\Feta{\iminduced F(G)}{\iminduced R(b)}=\Feta{\iminduced F(G)}{\iminduced R(a)}_{R_Y(u)}=F(Y)_{R_Y(u)}$$
The equivalence between the paths gives us $R_Y(u.v)=R_Y(u)R_{Y_u}(v)$, which concludes the proof of the translation invariance of $(F,R_\bullet)$.

Let us now consider the preservation of vertices, i.e. the bijectivity of $\iminduced{R}$. The graphs $G$ and $\iminduced F(G)$ both have a symmetric role in the construction of $\iminduced{R}$. Since $\iminduced{R}$ has been proved deterministic in Lemma \ref{lem:soundness_name_map}, $\iminduced R$ is also injective.
 Let $u'.t'\in V(\iminduced {F} (G))$ hold. Since $\iminduced{F}$ preserves the name $u'.t'$, there exists $u \in V(G)$ such that $u \intersectn u'.t' \neq \emptyset$. We can also note that $S'(u'.t') \in u'.t'\bint$. Therefore, there exists $t \in \bint$ such that $u.t=S'(u'.t')$. But as we noted earlier, $u.t$ is a prefix of $S(u.t)$ and therefore $u \intersectn u'.t' \neq \emptyset$. By Lemma \ref{lem:sigmas} we have that $S'(u'.t')$ belongs to $S(u.\bint)$ which concludes the proof of the surjectivity of $\iminduced{R}$.
 
 The quiescence of the invisible matter comes from the fact that for all $u$, the invisible matter vertices $u.l$ and $u.r$ can be moved only via the application of $\nu_u$. Since $F$ is continuous and translation invariant, we can prove that $S$ and $S'$ are composed of a bounded number of $\nu$, which proves the quiescence of invisible matter. See Lemma \ref{lem:im_quiescent} for the formal proof.

Continuity is relatively technical to prove, so the formal proof has been left to the appendix, see Lemma \ref{lemma:induced_dynamic_continuity}. Intuitively, the quiescence on the invisible matter of $\iminduced{F}$ as well as the bound on the number of $\nu$ in $S$ and $S'$ assures us that the vertices are not moved too far to or from the invisible matter, which allows us to concentrate on the visible matter. For the latter, the continuity of $F$ is a direct consequence of the continuity of $\iminduced{F}$, since the name decomposition algorithm can be performed locally. \hfill$\square$
\end{proof}

\begin{theorem}
 An NCGD is invertible if and only if its induced invisible matter dynamics is invertible.
\end{theorem}

\begin{proof}
 First, let us prove that the induced dynamics of an invertible NCGD is invertible.
Let $\iminduced{F}$ be an NCGD. As proved in Theorem \ref{th:ncgd_to_imcgd}, $\iminduced R$ is bijective.
Let ${\iminduced F}^{-1}$ be the dynamics such that $F^{-1}(\Feta{G}{a})=\anonymous{({\iminduced{F}}^{-1}(G)^{\vartriangle},{\iminduced{R}}^{-1}(a))}$.
 We have the equalities :
 \begin{align*}
 F(F^{-1}(\Feta{G}{a}))=&F(\anonymous{({\iminduced{F}}^{-1}(G)^{\vartriangle},{\iminduced{R}}^{-1}(a))})\\=&
 \anonymous{(\iminduced{F}\circ{\iminduced{F}}^{-1}(G)^{\vartriangle},\iminduced{R}\circ {\iminduced{R}}^{-1}(a))}\\=&\Feta{G}{a}
\end{align*}

Similarly, we have :
 \begin{align*}
 F^{-1}(F(\Feta{G}{a}))=&F^{-1}(\anonymous{({\iminduced{F}}(G)^{\vartriangle},{\iminduced{R}}(a))})\\=&
 \anonymous{({\iminduced{F}}^{-1}\circ\iminduced{F}(G)^{\vartriangle},{\iminduced{R}}^{-1}\circ \iminduced{R}(a))}\\=&\Feta{G}{a}
\end{align*}
And so $F$ is bijective and $(F,R_\bullet)$ invertible.

 Assume now that $F$ is bijective.
 Let $G\neq H$ be such that $\iminduced{F}(G)=\iminduced{F}(H)$. For all $y \in V(G)$ we have the following equalities:
 $$F(\Feta{G}{{\iminduced R_G}^{-1}(y)})=\Feta{\iminduced{F}(G)}{y}=\Feta{\iminduced{F}(H)}{y}=F(\Feta{H}{{\iminduced R_H}^{-1}(y)}$$
 By injectivity of $F$ we have $\Feta{G}{{\iminduced R_G}^{-1}(y)}=\Feta{H}{{\iminduced R_H}^{-1}(y)}$, so there is a renaming $\renaming$ such that $G = \renaming H$.
 Then, using translation invariance we obtain $\iminduced{F}(H)=\iminduced{F}(G)= \iminduced{F}(\renaming H)= \renaming \iminduced{F}(H)$, and so for all $u \in V(F(H))=V(H)$, $\renaming(u)=u$ and $G= \renaming H = H$, which conclude the proof of the injectivity of $\iminduced{F}$.
 
Let $H \in \mathcal{W}$, and $a \in V(H)$ hold. Since $F$ is surjective, there exist $G \in \mathcal{W}$ and $b \in V(G)$ such that:
 $$\Feta{\iminduced F(G)}{{\iminduced R_G}(b)}=F(\Feta{G}{b})= \Feta{H}{a}$$
Therefore, there exists a renaming $\renaming$ such that $\renaming \iminduced F(G) = H$. By invariance by translation, we have $\iminduced F(\renaming G) = H$ and the surjectivity of $\iminduced{F}$.
This concludes the  proof of the bijectivity of $\iminduced{F}$. \hfill$\square$
 \end{proof}

%Unfortunately, because of the loss of information during the simulation by an IMCGD, it is not possible, at this stage, to conclude directly that invertible NCGDs are reversible. However, we believe that it is possible to enrich our simulation with naming functions and thus preserve reversibility. The role of a naming function would be to name the visible matter of graphs with invisible matter. This would allow us to prove that the simulation process is invertible, and that for any NCGD $\iminduced{F}$ and induced dynamics $(F,R_\bullet)$, $\iminduced{(F^{-1})}=(\iminduced{F})^{-1}$. For this reason, we conjecture that there is a generalization of the reversibility theorem applicable to NCGDs.

\section{Conclusion}\label{sec:conclusion}

\noindent {\em Summary of contributions.} Previous works had raised the question whether parallel reversible computation allows for the local creation/destruction of nodes. Three different negative answers had been given in \cite{ArrighiCGD,ArrighiIC,ArrighiRCGD}. But we just described three relaxed settings in which this is possible: Causal Graph Dynamics over fully-anonymized graphs (ACGD); over pointer graphs modulo with invisible matter reservoirs (IMCGD); and finally CGD over graphs whose vertex names are in the algebra of `everywhere infinite binary trees' (NCGD). For each of these formalisms, we proved non-preservation of vertices by implementing the Hasslacher-Meyer example \cite{MeyerLGA}---see Examples \ref{ex:AHM}, \ref{ex:IMHM}, \ref{ex:NHM}. We also proved that we still had the classic Cellular Automata (CA) result that invertibility (i.e. mere bijectivity of the dynamics) implies reversibility (i.e. the inverse is itself a CGD)---via compactness---see Theorem \ref{th:imrev} in the case of IMCGDs, and we conjecture this is also the case for ACGDs and NCGDs Conjectures in \ref{conj:nrev}, \ref{conj:arev}. The answer to the question of reversibility versus local creation/destruction is thus formalism-dependent to some extent. We proceeded to examine the extent in which this is the case, and were able to show that (Reversible) ACGD, IMCGD and NCGD directly simulate each other---see Theorems \ref{th:imcgd_to_cgd}, \ref{th:cgd_to_imcgd}, \ref{th:ncgd_to_imcgd}, \ref{th:imcgd_to_ncgd}. They are but three presentations, in different levels of details, of a single robust setting in which reversibility and local creation/destruction are reconciled. 
\changes{Notice that the simulation of the ACGDs by the IMCGDs may fail to hold for some partially-defined IMCGD, i.e. having domain over certain subsets ${\cal S}\subseteq{\cal X}_{\Sigma,\pi}$. We believe this happens whenever ${\cal S}$ is such that witnessing a local symmetry $X^r=X^r_u$, allows us to conclude that the entire graph is symmetric, i.e. $X=X_u$. We leave this as a conjecture.}

\noindent {\em Perspectives.} 
\changes{Generally speaking, the relationship between CGDs and graph subshifts \cite{ArrighiSubshifts} needs to investigated, for instance expansive graph subshifts ought to be spacetime diagrams of partially-defined CGDs, and partially defined CGDs may help formalise simulations between different graph subshifts.}
Now that we proved that reversibility and local creation/ destruction of nodes are indeed compatible, one can ask whether such dynamics can make a graph grow indefinitely. We have answered this question positively in \cite{TimeArrow} for a variant of the HM example---based upon the formalism of the present paper. Besides the intriguing physical interpretations of this result as `toy model provably featuring a time arrow without past hypothesis', this potentially opens a number of intriguing mathematical questions: {\em What growth rate are achievable by Reversible CGD? Can that the growth happen homogeneously across the graph? Are indefinitely growing dynamics typical, amongst Reversible CDG?}

Just like Reversible CA were precursors to Quantum CA \cite{SchumacherWerner,ArrighiUCAUSAL}, Reversible CGD have paved the way for Quantum CGD \cite{ArrighiQCGD}. Toy models where time-varying topologies are reconciled with quantum theory, are of central interest to the foundations of theoretical physics \cite{QuantumGraphity1,QuantumGraphity2}---as it struggles to have general relativity and quantum mechanics coexist and interact. The `models of computation approach' brings the clarity and rigor of theoretical CS to the table, whereas the `natural and quantum computing approach' provides promising new abstractions based upon `information' rather than `matter'. Quantum CGD \cite{ArrighiQCGD}, however, lacked the ability to locally create/destroy nodes---which is necessary in order to model physically relevant scenarios. We have fixed this in \cite{ArrighiQNT} based on the findings of the present paper.

\section*{Acknowledgements} 
The authors acknowledge Gilles Dowek and Simon Martiel for many enlightening discussions. \changes{The manuscript has also been considerably improved thanks to the considerable time and effort dedicated to us by two anonymous referees.} This project/publication was made possible through the support of the ID\# 62312 grant from the John Templeton Foundation, as part of the \href{https://www.templeton.org/grant/the-quantum-information-structure-of-spacetime-qiss-second-phase}{‘The Quantum Information Structure of Spacetime’ Project (QISS)}. The opinions expressed in this project/publication are those of the author(s) and do not necessarily reflect the views of the John Templeton Foundation.

\bibliography{biblio}

\begin{thebibliography}{10}

\bibitem{ArrighiCGD}
P.~Arrighi and G.~Dowek.
\newblock {Causal graph dynamics}.
\newblock In {\em Proceedings of ICALP 2012, Warwick, July 2012, LNCS}, volume
  7392, pages 54--66, 2012.

\bibitem{ArrighiUCAUSAL}
P.~Arrighi, V.~Nesme, and R.~Werner.
\newblock {Unitarity plus causality implies localizability}.
\newblock {\em J. of Computer and Systems Sciences}, 77:372--378, 2010.
\newblock QIP 2010 (long talk).

\bibitem{ArrighiIC}
Pablo Arrighi and Gilles Dowek.
\newblock Causal graph dynamics (long version).
\newblock {\em Information and Computation}, 223:78--93, 2013.

\bibitem{TimeArrow}
Pablo Arrighi, Gilles Dowek, and Am{\'e}lia Durbec.
\newblock A toy model provably featuring an arrow of time without past
  hypothesis.
\newblock In Torben~{\AE}gidius Mogensen and {\L}ukasz Mikulski, editors, {\em
  Reversible Computation}, pages 50--68, Cham, 2024. Springer Nature
  Switzerland.

\bibitem{ArrighiSubshifts}
Pablo Arrighi, Am{\'{e}}lia Durbec, and Pierre Guillon.
\newblock Graph subshifts.
\newblock In {\em 19th Conference on Computability in Europe, CiE 2023, Batumi,
  Georgia, July 24-28, 2023, Proceedings}, volume 13967 of {\em Lecture Notes
  in Computer Science}, pages 261--274. Springer, 2023.
\newblock \href {https://doi.org/10.1007/978-3-031-36978-0\_21}
  {\path{doi:10.1007/978-3-031-36978-0\_21}}.

\bibitem{ArrighiQNT}
Pablo Arrighi, Am{\'e}lia Durbec, and Matt Wilson.
\newblock Quantum networks theory.
\newblock {\em arXiv preprint arXiv:2110.10587}, 2021.

\bibitem{ArrighiCreation}
Pablo Arrighi, Nicolas Durbec, and Aur{\'e}lien Emmanuel.
\newblock Reversibility vs local creation/destruction.
\newblock In {\em Proceedings of RC 2019, LLNCS}, volume 11497, pages 51--66.
  Springer, 2019.
\newblock \href {https://arxiv.org/abs/1805.10330} {\path{arXiv:1805.10330}},
  \href {https://doi.org/10.1007/978-3-030-21500-2_4}
  {\path{doi:10.1007/978-3-030-21500-2_4}}.

\bibitem{ArrighiQCGD}
Pablo Arrighi and Simon Martiel.
\newblock Quantum causal graph dynamics.
\newblock {\em Phys. Rev. D}, 96:024026, Jul 2017.
\newblock URL: \url{https://link.aps.org/doi/10.1103/PhysRevD.96.024026}, \href
  {https://doi.org/10.1103/PhysRevD.96.024026}
  {\path{doi:10.1103/PhysRevD.96.024026}}.

\bibitem{ArrighiCayleyNesme}
Pablo Arrighi, Simon Martiel, and Vincent Nesme.
\newblock Cellular automata over generalized {C}ayley graphs.
\newblock {\em Mathematical Structures in Computer Science}, 28(3):340--383,
  2018.
\newblock \href {https://doi.org/10.1017/S0960129517000044}
  {\path{doi:10.1017/S0960129517000044}}.

\bibitem{ArrighiBRCGD}
Pablo Arrighi, Simon Martiel, and Simon Perdrix.
\newblock Block representation of reversible causal graph dynamics.
\newblock In {\em Proceedings of FCT 2015, Gdansk, Poland, August 2015}, pages
  351--363. Springer, 2015.

\bibitem{ArrighiRC}
Pablo Arrighi, Simon Martiel, and Simon Perdrix.
\newblock Reversible causal graph dynamics.
\newblock In {\em Proceedings of International Conference on Reversible
  Computation, RC 2016, Bologna, Italy, July 2016}, pages 73--88. Springer,
  2016.

\bibitem{ArrighiRCGD}
Pablo Arrighi, Simon Martiel, and Simon Perdrix.
\newblock Reversible causal graph dynamics: invertibility, block
  representation, vertex-preservation.
\newblock {\em Natural Computing}, 19(1):157--178, 2020.
\newblock Pre-print arXiv:1502.04368.

\bibitem{Bartholdi}
L.~Bartholdi.
\newblock Gardens of {E}den and amenability on cellular automata.
\newblock {\em Journal of the European Mathematical Society}, 12(1):241--248,
  2010.

\bibitem{BFHAmalgamation}
P.~Boehm, H.R. Fonio, and A.~Habel.
\newblock Amalgamation of graph transformations: a synchronization mechanism.
\newblock {\em Journal of Computer and System Sciences}, 34(2-3):377--408,
  1987.

\bibitem{Chalopin}
J{\'e}r{\'e}mie Chalopin, Shantanu Das, and Peter Widmayer.
\newblock Deterministic symmetric rendezvous in arbitrary graphs: overcoming
  anonymity, failures and uncertainty.
\newblock In {\em Search Theory}, pages 175--195. Springer, 2013.

\bibitem{Danos200469}
V.~Danos and C.~Laneve.
\newblock Formal molecular biology.
\newblock {\em Theoretical Computer Science}, 325(1):69 -- 110, 2004.
\newblock Computational Systems Biology.
\newblock URL:
  \url{http://www.sciencedirect.com/science/article/pii/S0304397504002336},
  \href {https://doi.org/http://dx.doi.org/10.1016/j.tcs.2004.03.065}
  {\path{doi:http://dx.doi.org/10.1016/j.tcs.2004.03.065}}.

\bibitem{Durand-LoseBlock}
J.~O. Durand-Lose.
\newblock {Representing reversible cellular automata with reversible block
  cellular automata}.
\newblock {\em Discrete Mathematics and Theoretical Computer Science}, 145:154,
  2001.

\bibitem{EhrigLowe}
H.~Ehrig and M.~Lowe.
\newblock Parallel and distributed derivations in the single-pushout approach.
\newblock {\em Theoretical Computer Science}, 109(1-2):123--143, 1993.

\bibitem{CeccheriniEden}
Tullio Ceccherini-Silberstein;~Francesca Fiorenzi and Fabio Scarabotti.
\newblock {The Garden of Eden Theorem for cellular automata and for symbolic
  dynamical systems.}
\newblock In {\em {Random walks and geometry. Proceedings of a workshop at the
  Erwin Schr\"odinger Institute, Vienna, June 18 -- July 13, 2001. In
  collaboration with Klaus Schmidt and Wolfgang Woess. Collected papers.}},
  pages 73--108. Berlin: de Gruyter, 2004.

\bibitem{Gromov}
M.~Gromov.
\newblock {Endomorphisms of symbolic algebraic varieties}.
\newblock {\em Journal of the European Mathematical Society}, 1(2):109--197,
  April 1999.
\newblock URL: \url{http://dx.doi.org/10.1007/pl00011162}, \href
  {https://doi.org/10.1007/pl00011162} {\path{doi:10.1007/pl00011162}}.

\bibitem{QuantumGraphity2}
A.~Hamma, F.~Markopoulou, S.~Lloyd, F.~Caravelli, S.~Severini, K.~Markstrom,
  C.~Brouder, {\^A}.~Mestre, F.P. JAD, A.~Burinskii, et~al.
\newblock {A quantum Bose-Hubbard model with evolving graph as toy model for
  emergent spacetime}.
\newblock {\em Arxiv preprint arXiv:0911.5075}, 2009.

\bibitem{MeyerLGA}
Brosl Hasslacher and David~A. Meyer.
\newblock Modelling dynamical geometry with lattice gas automata.
\newblock Expanded version of a talk presented at the Seventh International
  Conference on the Discrete Simulation of Fluids held at the University of
  Oxford, June 1998.

\bibitem{Hedlund}
G.~A. Hedlund.
\newblock {Endomorphisms and automorphisms of the shift dynamical system}.
\newblock {\em Math. Systems Theory}, 3:320--375, 1969.

\bibitem{KariRevUndec}
J.~Kari.
\newblock {Reversibility of 2D cellular automata is undecidable}.
\newblock In {\em Cellular Automata: Theory and Experiment}, volume~45, pages
  379--385. MIT Press, 1991.

\bibitem{KariBlock}
J.~Kari.
\newblock {Representation of reversible cellular automata with block
  permutations}.
\newblock {\em Theory of Computing Systems}, 29(1):47--61, 1996.

\bibitem{KariCircuit}
J.~Kari.
\newblock {On the circuit depth of structurally reversible cellular automata}.
\newblock {\em Fundamenta Informaticae}, 38(1-2):93--107, 1999.

\bibitem{MeyerLove}
A.~Klales, D.~Cianci, Z.~Needell, D.~A. Meyer, and P.~J. Love.
\newblock Lattice gas simulations of dynamical geometry in two dimensions.
\newblock {\em Phys. Rev. E.}, 82(4):046705, Oct 2010.
\newblock \href {https://doi.org/10.1103/PhysRevE.82.046705}
  {\path{doi:10.1103/PhysRevE.82.046705}}.

\bibitem{QuantumGraphity1}
T.~Konopka, F.~Markopoulou, and L.~Smolin.
\newblock {Quantum graphity}.
\newblock {\em Arxiv preprint hep-th/0611197}, 2006.

\bibitem{Maignan}
Luidnel Maignan and Antoine Spicher.
\newblock Global graph transformations.
\newblock In {\em Proceedings of the 6th International Workshop on Graph
  Computation Models, L'Aquila, Italy, July 20, 2015.}, pages 34--49, 2015.

\bibitem{MartielMartin}
Simon Martiel and Bruno Martin.
\newblock Intrinsic universality of causal graph dynamics.
\newblock In Turlough Neary and Matthew Cook, editors, {\em {\rm Proceedings},
  Machines, Computations and Universality 2013, {\rm Z\"urich, Switzerland,
  9/09/2013 - 11/09/2013}}, volume 128 of {\em Electronic Proceedings in
  Theoretical Computer Science}, pages 137--149. Open Publishing Association,
  2013.
\newblock \href {https://doi.org/10.4204/EPTCS.128.19}
  {\path{doi:10.4204/EPTCS.128.19}}.

\bibitem{PapazianRemila}
C.~Papazian and E.~Remila.
\newblock Hyperbolic recognition by graph automata.
\newblock In {\em Automata, languages and programming: 29th international
  colloquium, ICALP 2002, M{\'a}laga, Spain, July 8-13, 2002: proceedings},
  volume 2380, page 330. Springer Verlag, 2002.

\bibitem{SchumacherWerner}
B.~Schumacher and R.~Werner.
\newblock {Reversible quantum cellular automata.}
\newblock arXiv pre-print quant-ph/0405174, 2004.

\bibitem{Sorkin}
R.~Sorkin.
\newblock {Time-evolution problem in Regge calculus}.
\newblock {\em Phys. Rev. D.}, 12(2):385--396, 1975.

\bibitem{TaentzerHL}
G.~Taentzer.
\newblock Parallel high-level replacement systems.
\newblock {\em Theoretical computer science}, 186(1-2):43--81, 1997.

\bibitem{TomitaSelfReproduction}
K.~Tomita, H.~Kurokawa, and S.~Murata.
\newblock Graph automata: natural expression of self-reproduction.
\newblock {\em Physica D: Nonlinear Phenomena}, 171(4):197 -- 210, 2002.
\newblock URL:
  \url{http://www.sciencedirect.com/science/article/pii/S0167278902006012},
  \href {https://doi.org/10.1016/S0167-2789(02)00601-2}
  {\path{doi:10.1016/S0167-2789(02)00601-2}}.

\end{thebibliography}

\appendix

\section{Formalism}\label{app:formalism}

This appendix provides formal definitions of the kinds of graphs we are using, together with the operations we perform upon them. None of this is specific to the reversible case; it can all be found in \cite{ArrighiCayleyNesme} and is reproduced here only for convenience.

\subsection{Graphs}\label{app:graphs}

\noindent Let $\pi$ be a finite set, $\Pi=\pi^2$, and ${\cal V}$ some universe of names.

\begin{definition}[Graph non-modulo]\label{def:graphs}
A {\em graph non-modulo} $G$ is given by 
\begin{itemize}
\item[$\bullet$] An at most countable subset $V(G)$ of ${\cal V}$, whose elements are called {\em vertices}.
\item[$\bullet$] A finite set $\pi$, whose elements are called {\em ports}.
\item[$\bullet$] A set $E(G)$ of non-intersecting two element subsets of $V(G):\pi$, whose elements are called edges. In other words an edge $e$ is of the form $\{u : a, v : b\}$, and $\forall e,e'\in E(G), e\intersectn e'\neq \emptyset \Rightarrow e=e'$. 
\item[$\bullet$] A partial function $\sigma$ from $V(G)$ to a finite set $\Sigma$;
\end{itemize}
The graph is assumed to be connected: for any two $u,v\in V(G)$, there exists $v_0,\ldots , v_{n}\in V(G)$, $a_0,b_0\ldots , a_{n-1},b_{n-1}\in \pi$ such that for all $i\in\{0\ldots n-1\}$, one has $\{v_i: a_i,v_{i+1}: b_i\}\in E(G)$ with $v_0=u$ and $v_n=v$.\\
The {\em set of graphs} with states in $\Sigma$ and ports $\pi$ is denoted by ${\cal G}_{\Sigma,\pi}$.
\end{definition}

We single out a vertex as the origin:
\begin{definition}[Pointed graph non-modulo]\label{def:pointedgraph}
A {\em pointed graph} is a pair $(G,p)$ with $p\in V(G)$. 
The {\em set of pointed graphs} with states in $\Sigma$ and ports $\pi$ is denoted by ${\cal P}_{\Sigma,\pi}$.
\end{definition}

The following notion expresses the idea that two graphs differ only up to the names of the vertices:
\begin{definition}[Isomorphism]\label{def:isomorphism}
An {\em isomorphism} $R$ is a function from ${\cal G}_{\pi}$ to ${\cal G}_{\pi}$ which is specified by a bijection $R(.)$ from ${\cal V}$ to ${\cal V}$. 
The image of a graph $G$ under the isomorphism $R$ is a graph $RG$ whose set of vertices is $R(V(G))$, and whose set of edges is $\{\{R(u):a,R(v):b\} \;|\; \{u:a,v:b\}\in E(G) \}$. 
Similarly, the image of a pointed graph $P=(G,p)$ is the pointed graph $RP=(RG,R(p))$. 
When $P$ and $Q$ are isomorphic we write $P\approx Q$, defining an equivalence relation on the set of pointed graphs. The definition extends to pointed labelled graphs.
\end{definition}
Pointed graph isomorphism rename the pointer in the same way as it renames the vertex upon which it points; which effectively means that the pointer does not move.

\begin{definition}[Pointed graphs modulo]\label{def:pointedmodulo}
Let $P$ be a pointed (labelled) graph $(G,p)$. The {\em pointed graph modulo} $X(P)$ is the equivalence class of $P$ with respect to the equivalence relation $\approx$. The {\em set of pointed graphs modulo} with ports $\pi$ is denoted by ${\cal X}_{\pi}$. The {\em set of labelled pointed Graphs modulo} with states $\Sigma$ and ports $\pi$ is denoted by ${\cal X}_{\Sigma,\pi}$.
\end{definition}

\subsection{Paths and vertices}\label{app:paths}%SECTION 2

Vertices of pointed graphs modulo isomorphism can be designated by a sequence of ports in $\Pi^*$ that leads, from the origin, to this vertex.

\begin{definition}[Path]\label{def:path}
Given a pointed graph modulo $X$, we say that $\alpha\in\Pi^*$ is a path of $X$ if and only if there is a finite sequence $\alpha=(a_i b_i)_{i\in\{0,...,n-1\}}$ of ports such that, starting from the pointer, it is possible to travel in the graph according to this sequence. More formally, $\alpha$ is a path if and only if there exists $(G,p)\in X$ and there also exists $v_0,\ldots , v_{n}\in V(G)$ such that for all $i\in\{0\ldots n-1\}$, one has $\{v_i: a_i,v_{i+1}: b_i\}\in E(G)$, with $v_0=p$ and $\alpha_i=a_ib_i$. Notice that the existence of a path does not depend on the choice of $(G,p)\in X$. The set of paths of $X$ is denoted by $V(X)$. 
\end{definition}
Paths can be seen as words on the alphabet $\Pi$ and thus come with a natural operation `$.$' of concatenation, a unit $\varepsilon$ denoting the empty path, and a notion of inverse path $\reversepath{\alpha}$ which stands for the path $\alpha$ read backwards, i.e for $\alpha=(a_ib_i)_{i\in\{0,...,n-1\}}$ we have $\alpha^{-1}=(b_ia_i)_{i\in\{n-1,...,0\}}$. Two paths are equivalent if they lead to same vertex:
\begin{definition}[Equivalence of paths]\label{def:equivpaths}
Given a pointed graph modulo $X$, we define the {\em equivalence of paths} relation $\equiv_{X}$ on $V(X)$ such that for all paths $\alpha,\alpha'\in V(X)$, $\alpha\equiv_{X} \alpha'$ if and only if, starting from the pointer, $\alpha$ and $\alpha'$ lead to the same vertex of $X$. 
More formally, $\alpha\equiv_{X}\alpha'$ if and only if there exists $(G,p)\in X$ and $v_1,\ldots , v_{n},v'_1,\ldots , v'_{n'}\in V(G)$ such that for all $i\in\{0\ldots n-1\}$, $i'\in\{0\ldots n'-1\}$, one has $\{v_i: a_i,v_{i+1}: b_i\}\in E(G)$, $\{v'_{i'}: a'_{i'},v'_{i'+1}: b'_{i'}\}\in E(G)$, with  $v_0=p$, $v'_0=p$, $\alpha=(a_ib_i)_{i\in\{0,...,n-1\}}$, $\alpha'=(a'_{i'}b'_{i'})_{i\in\{0,...,n'-1\}}$ and $v_{n}=v_{n'}$.
We write $\hat{\alpha}$ for the equivalence class of $\alpha$ with respect to $\equiv_X$.
\end{definition}

It is useful to undo the modulo, i.e. to obtain a canonical instance $(G(X),\varepsilon)$ of the equivalence class $X$.
\begin{definition}[Associated graph]\label{def:associatedgraph}
Let $X$ be a pointed graph modulo. Let $G(X)$ be the graph such that:
\begin{itemize}
\item[$\bullet$] The set of vertices $V(G(X))$ is the set of equivalence classes of $V(X)$;
\item[$\bullet$] The edge $\{\hat{\alpha}: a,\hat{\beta}: b\}$ is in $E(G(X))$ if and only if $\alpha.ab \in V(X)$ and $\alpha.ab\equiv_X \beta$, for all $\alpha\in \hat{\alpha}$ and $\beta\in \hat{\beta}$.
\end{itemize}
We define the {\em associated graph} to be $G(X)$.
\end{definition}

\noindent {\em Notations.} The following are three presentations of the same mathematical object:
\begin{itemize}
\item[$\bullet$] a graph modulo $X$,
\item[$\bullet$] its associated graph $G(X)$
\item[$\bullet$] the algebraic structure $\langle V(X),\equiv_X\rangle$
\end{itemize}
Each vertex of this mathematical object can thus be designated by
\begin{itemize}
\item[$\bullet$]  $\hat{\alpha}$ an equivalence class of $V(X)$, i.e. the set of all paths leading to this vertex starting from $\hat{\varepsilon}$,
\item[$\bullet$] or more directly by $\alpha$ an element of an equivalence class $\hat{\alpha}$ of $X$, i.e. a particular path leading to this vertex starting from $\varepsilon$.
\end{itemize}
These two remarks lead to the following mathematical conventions, which we adopt for convenience:
\begin{itemize}
\item[$\bullet$] $\hat{\alpha}$ and $\alpha$ are no longer distinguished unless otherwise specified. The latter notation is given the meaning of the former. We speak of a ``vertex'' $\alpha$ in $V(X)$.
\item[$\bullet$] It follows that `$\equiv_X$' and `$=$' are no longer distinguished unless otherwise specified. The latter notation is given the meaning of the former. I.e. we speak of ``equality of vertices'' $\alpha=\beta$ (when strictly speaking we just have $\hat{\alpha}=\hat{\beta}$).
\end{itemize}

\subsection{Operations over pointed Graphs modulo}\label{app:operationsmodulo}

\noindent {\em Sub-disks}. For a pointed graph $(G,p)$ non-modulo:
\begin{itemize}
\item[$\bullet$] the neighbours of radius $r$ are just those vertices which can be reached in $r$ steps starting from the pointer $p$;
\item[$\bullet$] the disk of radius $r$, written $G^r_p$, is the subgraph induced by the neighbours of radius $r+1$, with labellings restricted to the neighbours of radius $r$, and pointed at $p$.
\item[$\bullet$] this is readily extended to a set $V\subseteq V(G)$ with $p\in V$, i.e. $G_V^r$ is the subgraph induced by the neighbours of radius $r+1$ of the vertices in $V$, with labellings restricted to the neighbours of radius $r$, and pointed at $p$.
\end{itemize}
For a graph modulo, on the other hand, the analogous operation is:
\begin{definition}[Disk]
Let $X\in {\cal X}_{\Sigma,\pi}$ be a pointed graph modulo and $G$ its associated graph. 
Let $X^r$ be $X(G^r_\varepsilon)$.
The graph modulo $X^r\in {\cal X}_{\Sigma,\pi}$ is referred to as the {\em disk of radius $r$} of $X$. The {\em set of disks of radius $r$} with states $\Sigma$ and ports $\pi$ is denoted by ${\cal X}^r_{\Sigma,\pi}$.
\end{definition}

\begin{definition}[Size]\label{def:size}
Let $X\in {\cal X}_{\Sigma,\pi}$ be a pointed graph modulo. We say that a vertex $u\in V(X)$ has size less or equal to $r+1$, and write $|u|\leq r+1$, if and only if $u\in V(X^r)$.  
\end{definition}

\noindent {\em Shifts} just move the pointer vertex:
\begin{definition}[Shift]\label{def:shift}
Let $X\in {\cal X}_{\Sigma,\pi}$ be a pointed graph modulo and $G$ its associated graph. 
Consider $u\in V(X)$ or $X^r$ for some $r$, and consider the pointed graph $(G,u)$, which is the same as $(G,\varepsilon)$ but with a different pointer. Let $X_u$ be $X(G,u)$. The pointed graph modulo $X_u$ is referred to as {\em $X$ shifted by $u$}.\\ 
\end{definition}

\subsection{Compactness of the Invisible Extension}\label{subsec:IMCGDcompactness}

The main result of this subsection is that, although ${\cal Y}$ is not a compact subset of ${\cal X}'$ by itself, IMCGD can be extended continuously over the  compact closure of ${\cal Y}$ in ${\cal X}'$.\\
Indeed ${\cal Y}$ is not a compact subset of ${\cal X}'$, for instance the sequence $(Y_{\bdot(lm)^r})_{r \in \N}$, pointing ever further into the invisible matter, has no convergent subsequence in ${\cal Y}$ but has one in $\mathcal{X}'$.

\begin{proposition}[Closure of invisible]\label{closure no visible point}
$Y'$ is in $\overline{{\cal Y}}$ and has no visible matter if and only if there exists $(u_n)$ a sequence of path in $\{lm,rm\}^n$ such that $u_n$ is a suffix of $u_{n+1}$, $Y'^n = T_{u_n}^n$, and 
$$Y' = \bigcup \limits_{n = 1}^\infty \nearrow T_{u_n}$$ i.e. $Y'$ is the non-decreasing union of the $(T_{u_n})$. %Conversely, consider $(u_n)$ a sequence of paths in in $\{lm,rm\}^n$ such that $u_n$ is a suffix of $u_{n+1}$. Then, 
%$$Y' = \bigcup \limits_{n = 1}^\infty \nearrow T_{u_n}$$
%is in $\overline{{\cal Y}}$ and has no visible matter.
\end{proposition}

\begin{proof}

First notice that $T_v$ is a sub-graph of $T_{uv}$. Indeed, by definition of $T$, the vertex $u$ in $T$ is the root of a copy of $T$, thus $T$ is a subgraph of $T_u$. Shifting this statement by $v$, $T_v$ is a sub-graph of $T_{uv}$. Thus it makes sense to speak about their non-decreasing union.\\
Next, for any $Y$, $Y = \bigcup \limits_{n = 1}^\infty \nearrow Y^n$. So if $Y'$ is a graph of $\overline{{\cal Y}}$ with no point in the visible matter, then
$$ Y' = \bigcup \limits_{n = 1}^\infty \nearrow Y'^n = \bigcup \limits_{n = 1}^\infty \nearrow T_{u_n}^{n}.$$
Conversely, any such non-decreasing union is equal to  $\lim_n (Z_{\bdot u_n})$, for any graph $Z$ of ${\cal X}$ completed into an element of ${\cal Y}$, thus it belongs to ${\cal \overline{Y}}$.
\hfill$\square$ \end{proof}

\begin{lemma}[Visible starting paths]\label{lem:vispaths}
Consider $Y'$ in $\overline{{\cal Y}}$ with $\varepsilon$ visible, and $v$ in $V(Y')$. Then $v$ can be decomposed as $u\bdot t$, with $u$ in $\Pi^*$, and $t$ in $\{lm,rm\}^*$ if and only if $v$ is invisible. Moreover, for any $s,t$ in $\{lm,rm\}^*$, we have that $\bdot s,\bdot t$ are in $Y'$, and $\bdot s\equiv_{Y'} \bdot t$ if and only if $s=t$.
\end{lemma}
\begin{proof} First consider $Y\in {\cal Y}$ with $\varepsilon$ visible. Clearly $\bdot s$ is in $Y$ and is the minimal path to $\bdot s$, as the invisible matter is tree. If $v$ is invisible then $v$ has only one node $u$ connecting its invisible matter tree to the visible matter, $v$ can be minimally decomposed into $u\bdot t$ with $t$ in $\{lm,rm\}^*$.\\
The same holds in the closure. Indeed consider $Y'$ in $\overline{{\cal Y}}$ with $\varepsilon$ visible. Let $n=\max(|v|,|s|+1,|t|)$ and pick $Y$ in ${\cal Y}$ such that $Y'^{n} = Y^{n}$. By definition of $|v|$ we have that $p$ is a shortest path from $\varepsilon$ to $v$ in $Y$ if and only if it is one in $Y'$. Therefore the form of the decomposition of $v$, and its invisibility when $t$ is not empty, carry through to $Y'$. So does the existence of $\bdot s$. Finally, we have that $\bdot s\equiv_{Y'} \bdot t$ implies $\bdot s\equiv_{Y} \bdot t$, which implies $s=t$, due to the tree structure of the invisible matter in $Y$. \hfill$\square$
%	\begin{itemize}
%		\item[$\bullet$] If, in $Y$, $\varepsilon$ and $v$ are on the same invisible matter tree, then, the shortest path from $\varepsilon$ to $v$ is to go up to their most recent common ancestor and back down to $v$, which can be written $v = \overline{s}t$, for some $s$ and $t$ in $\{lm,rm\}^*$, in both $Y'$ and $Y$ since it's a shortest path.
%		\item[$\bullet$] If, in $Y$, $\varepsilon$ and $v$ are on the same invisible matter tree, then by Remark \ref{Invisible Matter Graph structure} there exists a $u$ in $\pi^*$, and an $s$ and a $t$ in $\{lm,rm\}^*$ such that $v$ equals $\overline{s}.u.t$, $u.t$, $\overline{s}.u$, $u$ or $\overline{s}t$ in $Y$. This decomposition describes the shortest path from $\varepsilon$ to $v$, so it still holds in $Y'$.
%	\end{itemize}
\end{proof}

\begin{lemma}[Invisible starting paths]\label{lem:invpaths}
Consider $Y'$ in $\overline{{\cal Y}}$ with $\varepsilon$ invisible, and $v$ in $Y'$. Then $v$ can be decomposed as $\reversepath{s}\bdot u$, in which case $v$ is visible, or as $\reversepath{s}\bdot u\bdot t$, in which case $v$ is invisible---with $u$ in $\Pi^*$, and $s,t$ in $\{lm,rm\}^*$. Moreover, any $s,t$ in $\{lm,rm\}^*$, are also in $Y'$, and we have that $s\equiv_{Y'}t$ if and only if $s=t$.
\end{lemma}
\begin{proof}
Same proof scheme as in Lemma \ref{lem:vispaths}.  
\hfill$\square$ \end{proof}

\begin{lemma}[Finite invisible root]\label{lem:finiteinvroot}
Consider $Y'$ in  $\overline{{\cal Y}}$. If $Y'$ has no visible matter, then, for all $n$, there exists a unique word $u_n$ in $\{lm,rm\}^n$ such that $Y^n = T_{u_n}^n$. As a consequence, $\reversepath{u}_n$ is the unique word in $\{ml,mr\}^n$ such that $\reversepath{u}_n$ is in $Y'$. Moreover, if $p \leq n$, then $u_p$ is a suffix of $u_n$.
\end{lemma}

\begin{proof}
Consider $Y'$ in  $\overline{{\cal Y}}$. Pick $Y$ in ${\cal Y}$ such that $Y^n = Y'^n$. Since $Y'$ has no visible vertex, $Y' = Y_{\bdot s}$ with $|s| > n$, and we can take $u$ to be the suffix of length $n$ of $s$, and $t$ the complementary prefix, such that $s = tu$. $Y'^n$ is included in the invisible matter tree rooted in $Y_{\bdot t}$, hence $Y^n = Y'^n = (Y_{\bdot t})_{u}^n = T_{u}^n$.\\ For uniqueness, notice that for any two words $u,v$ of length $n$, $T_{u}^n = T_{v}^n$ implies $u = v$.\\
Since $\reversepath{u}_n$ is the only word of length $n$ in $\{ml,mr\}^n$ to represent a valid path of $Y'$, its prefix of length $p$ is the only word of length $p$ in $\{ml,mr\}^p$ to represent a valid path of $Y'$, which we know is $\overline{u}_p$.
\hfill$\square$ \end{proof}

\begin{proposition}[Invisible matter extension]\label{prop:IMextension_app}
Consider $(F,R_\bullet)$ a continuous and shift-invariant dynamics over ${\cal Y}$. We have that $(F,R_\bullet)$ is invisible matter quiescent if and only if $(F,R_\bullet)$ can be continuously extended to ${\cal \overline{Y}}$ by letting $F(T_u) = T_u$ and $R_{T_u} = Id$ for any $u$ in $\{l,r\}^{-\N}$.
\end{proposition}

\begin{proof}
Notice how, for all $Y$ and $u_{k+1}\in \{lm,rm\}^{k+1} \cap V(Y)$, we have that $Y_{u_{k+1}}^k = T_u^k$.\\
$[\Rightarrow]$. Let $(F,R_{\bullet})$ be a continuous, shift-invariant and invisible matter quiescent dynamics. Take  $u$ a  left-infinite word in $\{lm,rm\}^{- \N}$. 
Continuity of $F$ over ${\cal Y}$ states that for all $m$ there is an $n\geq m$ such that $(F(Y))^m=(F(Y^n))^m$. By invisible matter quiescence there is a $b$ such that for all $Z, p\in \{lm,rm\}^b, v \in \{lm,rm\}^*$, $R_Z(\bdot p v)=R_Z(\bdot p)R_{Z_{\bdot p}}(v)=R_Z(\bdot p )v$. Combining these, $(F(T_u^n))^m=(F(Z_{\bdot u_{b+n+1}}^n))^{m}=(F(Z_{\bdot u_{b+n+1}}))^{m}=(F(Z)_{R_Z({\bdot u_b})u_{n+1}})^{m}=T_u^m$. Hence, if we extend $F$ to ${\cal \overline{Y}}$ by $F(T_u)=T_u$, we get $(F(T_u))^m=(F(T_u^n))^m$, and so $F$ remains continuous. Similarly, continuity of $R_\bullet$ over ${\cal Y}$ states that for all $m$ there is an $n\geq m$ such that $R_Z^m=R_{Z^n}^m$. Again combining it with invisible matter quiescence, $R_{T_u^n}^m=R_{Z_{\bdot u_{b+n+1}}^n}^m=Id$. Hence, if we extend $R$ to ${\cal \overline{Y}}$ by $R_{T_u}=Id$, we get $R_{T_u}^m=R_{T_u^n}^m$, and so $R_\bullet$ remains continuous.\\
We can thus continuously extend $(F,R_{\bullet})$ by setting $F(T_u) = T_u$ and $R_{T_u} = Id$.\\
$[\Leftarrow]$. Conversely, no longer assume invisible matter quiescence, and suppose instead that $(F,R_{\bullet})$, when extended by $F(T_u) = T_u$ and $R_{T_u} = Id$, is continuous over $\overline{{\cal Y}}$. Since $\overline{{\cal Y}}$ is compact, $(F,R_{\bullet})$ is uniformly continuous (by the Heine–Cantor Theorem). Take $c$ such that for all $Y$, $a$ in $Y$, and $|a|$=1, we have $|R_Y(a)|\leq c$. Such a $c$ exists by Lemma 3 of \cite{ArrighiCayleyNesme}. Take $b$ such that, for all $Y$, we have
$$\bullet\ (F(Y))^c = (F(Y^b))^c\quad \bullet\ \dom\,{R_Y}^c\subseteq V(Y^b)\quad\bullet\ {R_Y}^c = {R_{Y^b}}^c $$
We prove, by induction, that $b+1$ is the bound for invisible matter quiescence. Indeed, our inductive hypothesis is that for all $p$ in $\{lm,rm\}^{b+1}$ and $w$ in $\{lm,rm\}^n$, we have $R_{Y_{\bdot p}}(w) = w$. The hypothesis holds for $n=0$, because a consequence of shift-invariance is that $R_Y(\varepsilon)=\varepsilon$ for any $Y$. Suppose it holds for some $n$. Take $a$ in $\{lm,rm\}$. We have $R_{Y_{\bdot p}}(wa) = R_{Y_{\bdot p}}(w) R_{Y_{\bdot pw}}(a)$. Since $|p|=b+1$, we have $Y_{\bdot pw}^b=T_{upw}^b$ for any left-infinite $u$ in $\{lm,rm\}^{- \N}$. By the choice of $c$ and $b$, $R_{Y_{\bdot pw}}(a) = R_{Y_{\bdot pw}}^c(a) = R_{Y_{\bdot pw}^b}^c(a) = R_{T_{upw}^b}^c(a) = R_{T_{upw}}^c(a) = Id^c(a) = a $. Putting things together, we have $R_{Y_{\bdot p}}(wa) = wa$.
\hfill$\square$ \end{proof}

\begin{theorem}
An IMCGD can be extended into a vertex-preserving invisible matter quiescent CGD over $\overline{\cal Y}$.
\end{theorem}
\begin{proof}
Consider $(F,R_\bullet)$ an IMCGD. Extend it to $\overline{\cal Y}$ by setting $F(T_u) = T_u$ and $R_{T_u} = Id$. By Th.~\ref{th:IMextension} the extension is still continuous. It is still vertex-preserving since $R_{T_u}$ is bijective. Therefore it is still bounded. It is still shift-invariant since $R_{T_{u}}(vw)=R_{T_{u}}(v)R_{T_{uv}}(w)=vw$. 
\hfill$\square$ \end{proof}

\begin{corollary} \label{cor:IMQ_compactness}
An IMCGD $(F,R_{\bullet})$ is {\em uniformly continuous}. I.e.  for all $m$, there exists $n$, such that for any $Y$, 
$$\bullet\ (F(Y))^m=(F(Y^n))^m\qquad \bullet\ \dom\,R_{Y}^m\subseteq V(Y^n)\textrm{ and }R_{Y}^m=R_{Y^n}^m$$
where $R_{Y}^m$ denotes the partial map obtained as the restriction of $R_Y$ to the co-domain $(F(Y))^m$, using the natural inclusion of $(F(Y))^m$ into $F(Y)$.
\end{corollary}
\begin{proof}
Extend the IMCGD to the compact metric space $\overline{\cal Y}$ and apply Heine's Theorem to find that continuity implies uniform continuity.
\hfill$\square$ \end{proof}

\section{IMCGDs by NCGDs}

\begin{lemma}\label{lem:soundness_name_map}
For all $u.t \in V(G).\bint$, there is at most one $v$ and $s$ such that $S(u.t)=S'(v.s)$ where $S$ and $S'$ results from the application of Lemma \ref{lem:sigmas}.
\end{lemma}

\begin{proof}
 Let us suppose there exists $v_1.s_1$ and $v_2.t_2$ such that $S_1(u.t)=S_1'(v_1.s_1)$ and $S_2(u.t)=S'_2(v_2.s_2)$. By construction of $S_1$ and $S_2$, we have that there exists $t'\in \bint$ such that $S_1(u.t)t'=S_2(u.t)$ or $S_1(u.t)=S_2(u.t)t'$. Without loss of generality suppose that $S_1(u.t)t'=S_2(u.t)$, therefore we have the following equalities: 
 $$S'_1(v_1.s_1)t'=S_1(u.t)t'=S_2(u.t)=S'_2(v_2.s_2)$$
 Then, remark that for all $u$, for all $S$, $u.t$ is a prefix of $S(u.t)$, so we can rewrite the precedent equality as $v_1.s_1s_1't'= v_2.s_2s_2'$ with some $s_1',s_2' \in \bint$. As $F(G)$ is a well-named graph, this implies $v_1=v_2$. $S_1=S_2$ and $S_1'=S_2'$ which gives us that $S_1'(v_1.s_1)=S_1'(v_1.s_2)$. For all $u$, $\nu_u$ is injective, so $S_1$ is injective and we have $v_1.s_1=v_1.s_2$. Again, $F(G)$ is a well-named graph, so $s_1=s_2$.
\hfill$\square$ \end{proof}
\begin{lemma}\label{lem:im_quiescent}
 The induced dynamics of an NCGD is invisible matter quiescent.
\end{lemma}

\begin{proof}
We want to prove there exists a bound $b$ such that, for all $Y_{\bdot s} \in {\mathcal Y}$, and for all $t$ in $\{lm,rm\}^*$, we have $|s|\geq b \implies R_{Y_{\bdot s}}(t) = t $. For $R$ induced by $\overline R$ this is implied by: for all $G$, for all $v\in V(G)$ and $s,t$ in $\bint$ then $|s|\geq b$ implies $\overline{R}(v.st)= \overline R(v.s)t$.

First, let us prove there exists a bound $b \in \N$ such that for all $v \in V(G)$ and $v' \in V(F(G))$, $v \intersectn v' \neq \emptyset$ implies there exists $r,r' \in \bint$ such that $|r|\leq b$, $|r'| \leq b$ and $S(v.r)=S'(v'.r')$. 

This comes from the fact that $\overline{R}$ is $S'^{-1}\circ S$ with $S$ and $S'$ finite compositions of sigmas, whose overall length can be bounded by $b$ as proven in Lem. \ref{lem:sigmas}. Indeed, $S$ and $S'$ are computed from disks of radius $r$, and we have proven in the soundness of Def. \ref{def:induced_dyn} that the length of $S$ and $S'$ is invariant under renamings. Because there is a bounded number of disks of radius $r$, we take $b$ to be the maximum of these lengths. 

By definition of $S$ and $S'$, we also have that for all $t\in \bint$, $S(v.rt)=S'(v'.r't)$ and so $\overline R(v.rt)=v'.r't$.

Let $v \in V(G)$, and $s \in \bint$ such that $|s|\geq b$. As $F$ is name-preserving, there exists $v'$ such that $v.s \intersectn v' \neq \emptyset$. As stated above, there exists $r,r'\in \bint$ such that $s=rs_2$, $|r| \leq b$, $|r'| \leq b$, $S(v.r)=S'(v'.r')$ and so $\overline R(v.r)=v'.r'$. But we also have that $S(v.rs_2t)=S'(v'.r's_2t)$, therefore $\overline{R}(v.st)= \overline{R}(v.rs_2t)=v'.r's_2t=\overline{R}(v.r)s_2t$. For $t=\varepsilon$, we have that $\overline{R}(v.s)=\overline{R}(v.r)s_2$, which gives us to $\overline{R}(v.st)=R(v.r)s_2t=\overline{R}(v.s)t$.
\hfill$\square$ \end{proof}

\begin{lemma}\label{lemma:induced_dynamic_continuity}
 The induced dynamics of an NCGD is continuous.
\end{lemma}

\begin{proof}
First we prove that for all $Y \in \mathcal{Y}$, for all $m \in \N$ there exists $n \in \N$ such that $(F(Y))^m = (F(Y^n))^m$.
Let $Y \in \mathcal{Y}$, $G \in \mathcal W$ and $v \in V(G)$ such that $Y=\Feta{G}{v}$. Notice the following equalities: $$ \anonymous{(G^\vartriangle,v)^m}= \anonymous{((G^\vartriangle)^m_v,v)} = \anonymous{((G^{m\vartriangle}_v)^m_v,v)}$$
We have by definition of $F$ that $(F(Y))^m=F\anonymous{(G^\vartriangle,v)^m}=\anonymous{(F(G)^\vartriangle,R(v))^m}$. 
Using the precedent remark, we have that $(F(Y))^m=\anonymous{((\overline (F(G))^{m\vartriangle}_v)^m_v,v)}$. By continuity of $\overline F$, for $v'=v$ we have that there exists $n\in \N$ such that: $$(F(Y))^m=\anonymous{((\overline{F}(G^n_v))^{m\vartriangle}_v)^m_v,\overline R_G(v))}=\anonymous{(\overline F(G^n_v)^\vartriangle),\overline R_G(v))^m}=(F(Y^n))^m$$

Now let us focus on proving that for all $m \in \N$ there exists $n \in \N$, such that $\text{dom } \overline R^m_{\Feta{G}{v}} \subseteq V(\anonymous{(G^\vartriangle,v)^n})$ and $\overline R^m_{\Feta{G}{v}}=\overline R^m_{\anonymous{(G^\vartriangle,v)^n}}$.

Let $u.s \in V(G).\bint$ and $u'.s' \in V(G)$ such that $R(u.s) = u'.s'$. $S(u.s)=S'(u.s')$ so there exists $t,t' \in \bint$ such that $u.st=u'.s't'$. As $u \intersectn u' \neq \emptyset$ and by continuity of $\overline F$ we have that for all $m\in \N$ there exists $n \in \N$ such that if $u' \in (\overline F(G))^m_{v'}$ then $u \in V(G^n_v)$.

Let $b$ the invisible matter quiescence bound. As stated in Lemma \ref{lem:im_quiescent}, $S$ and $S'$ are a bounded composition of sigmas, therefore for all $b \in \N$ there exists $b'\in \N$ such that $|s|\leq b$ implies $|s'|\leq b'$. If $|s|>b$, then there exists $s_1,s_2 \in \bint$ such that $s=s_1s_2$, $|s_1|\leq b$ and $R(u.s_1s_2)=R(u.s_1)s_2$. Let $s_1'$ such that $u.s_1 = u'.s_1'$, then $s_1'\leq b'$. We have :
$$u's'=R(u.s)=R(u.s_1s_2)=R(u.s_1)s_2=u'.s_1's_2$$ $|s_1'|\leq b'$ and $s_2 \leq |s|$ therefore $|s'|\leq b'+|s|$. Summarizing, we have that for all $d$, there exists a $d'$ such that $|s| \leq d$ implies $|s'| \leq d'$. But the construction of $R$ is symmetrical, therefore we also have that for all $d'$, there exists $d \in \N$ such that $|s'| \leq d'$ implies $|s| \leq d$. 

Let $u.s \in V(G)$ and $u'.s' \in V(G)$ such that $R(u.s) = u'.s'$. If $u.s \in V(G^n_v)$ then $u \in V(G^n_v)$ as $v$ is in the visible matter and $|s|\leq n$. As stated above, there exists a bound $m$ such that $u \in V{ \anonymous{(G^\vartriangle,v)^m}}$ and there is a bound $d$ such that $|s| \leq d$ therefore $u.s \in V(\anonymous{(G^\vartriangle,v)^{m+d}})$.
Because $u \in V(\anonymous{(G^\vartriangle,v)^{m+d}})$, and $\overline R(u.s)$ is only computed from $u$ we also have that $\overline R_{\Feta{G}{v}}(u.s)=\overline R_{\anonymous{(G^\vartriangle,v)^n}}(u.s)$. 

\hfill$\square$ 
\end{proof}

\end{document}